\documentclass[a4paper,twocolumn,11pt,unpublished]{quantumarticle}
\pdfoutput=1
\usepackage[utf8]{inputenc}
\usepackage[english]{babel}
\usepackage[T1]{fontenc}
\usepackage{amsmath,amssymb,amsthm}
\usepackage{hyperref,cleveref}
\usepackage{dsfont}
\usepackage{mathtools}
\usepackage {cancel}
\usepackage{comment}
\newcommand{\updown}{\hspace{-0.8mm}\uparrow\downarrow}
\newcommand{\downup}{\hspace{-0.8mm}\downarrow\uparrow}

\newcommand{\upket}{\hspace{1mm}|\hspace{-1.6mm}\uparrow\hspace{0.5mm}\rangle}
\newcommand{\downket}{|\hspace{-1.6mm}\downarrow\hspace{0.5mm}\rangle}

\newcommand\ketGS{|\Psi_0\rangle} 
\newcommand\braGS{\langle \Psi_0 |}

\newcommand{\ket}[1]{|#1\rangle}
\newcommand{\bra}[1]{\langle#1|}
\newcommand{\ketbra}[2]{\ket{#1}\bra{#2}}

\usepackage[numbers,sort&compress]{natbib}

\usepackage{tikz}
\usepackage{lipsum}
\usepackage{mathrsfs}

\newtheorem{lemma}{Lemma}
\newtheorem{fact}{Fact}

\begin{document}
\title{Rapidly mixing loop representation quantum Monte Carlo for Heisenberg models on star-like bipartite graphs}

\author{Jun Takahashi}
\affiliation{Center for Quantum Information and Control, University of New Mexico, Albuquerque, NM 87131, USA}
\affiliation{The Institute for Solid State Physics, The University of Tokyo, Kashiwa, Chiba, Japan}
\orcid{0000-0003-1327-9322}
\email{juntakahashi@issp.u-tokyo.ac.jp}
\author{Sam Slezak}
\affiliation{Information Sciences, Los Alamos National Laboratory, Los Alamos, NM, USA}
\author{Elizabeth Crosson}
\affiliation{Unaffiliated}
\maketitle

\begin{abstract}
Quantum Monte Carlo (QMC) methods have proven invaluable in condensed matter physics, particularly for studying ground states and thermal equilibrium properties of quantum Hamiltonians without a sign problem. Over the past decade, significant progress has also been made on their rigorous convergence analysis. 
 Heisenberg antiferromagnets (AFM) with bipartite interaction graphs are a popular target of computational QMC studies due to their physical importance, but despite the apparent empirical efficiency of these simulations it remains an open question whether efficient classical approximation of the ground energy is possible in general.  In this work we introduce a ground state variant of the stochastic series expansion QMC method, and for the special class of AFM on interaction graphs with an $O(1)$-bipartite component (star-like), we prove rapid mixing of the associated QMC Markov chain (polynomial time in the number of qubits) by using Jerrum and Sinclair's method of canonical paths.  This is the first Markov chain analysis of a practical class of QMC algorithms with the loop representation of Heisenberg models. 
 Our findings contribute to the broader effort to resolve the computational complexity of Heisenberg AFM on general bipartite interaction graphs.

\end{abstract}

\section{Introduction}
 The Heisenberg Hamiltonian couples neighboring spin-1/2 degrees of freedom arranged on fixed spatial sites by the \emph{exchange interaction} of the form $\pm \vec{\sigma}_i\cdot \vec{\sigma}_j = \pm (X_i X_j + Y_i Y_j + Z_i Z_j)$, which models the effective energetic tendencies of neighboring electron spins in molecular orbitals within a solid ($i,j$ refer to neighboring sites, and $\vec{\sigma} = (X,Y,Z)$ is a vector of Pauli matrices).   The coupling $-\vec{\sigma}_i\cdot \vec{\sigma}_j$ in a Hamiltonian is called \emph{ferromagnetic} (FM) because it energetically favors alignment of spins, while the coupling $+\vec{\sigma}_i\cdot \vec{\sigma}_j$ favors anti-alignment and is called \emph{antiferromanetic} (AFM).

Exact solutions for Heisenberg models are known only for a few cases \cite{lieb1961two,bet31zur,sha81exa,maj69nex} so approximate or computational solutions are needed to study the model in more complicated geometries such as 2D lattices and beyond. 
Quantum Monte Carlo (QMC) methods seek computational solutions by reducing the approximation of equilibrium observables in certain quantum spin systems to the problem of sampling classical configurations from a weighted probability distribution.  
QMC methods use a Markov chain to sample from the target distribution, and Monte Carlo estimation to approximate observables of interest.   An important condition for the efficient application of QMC is that the model should not have a \emph{sign problem} \cite{Loh1990SignProb,Marshall1955Antiferro} 
, meaning that there should be some choice of local basis in which all of the off-diagonal Hamiltonian matrix elements are real and non-positive. 
In both Suzuki's original formulation of QMC \cite{Suzuki1977MonteCarlo,Suzuki1976Relationship} and more modern QMC treatments for spin systems \cite{sandvik2010computational}, this condition ensures that the paths contributing to the path integral for the ground state all have non-negative weights. 
More generally, the absence of a sign problem reduces the computational complexity of estimating equilibrium observables \cite{cubitt2016complexity} and such models without a sign problem are now commonly called ``stoquastic'' in the context of Hamiltonian complexity \cite{bravyi2007complexity,bravyi2008complexitySFF,bravyi2015monte}. 

Ferromagnetic Heisenberg models are always stoquastic (the interaction $-(X_i X_j + Y_i Y_j + Z_i Z_j)$ has all real and non-positive matrix elements in the $Z$-basis).   For the antiferromagnetic interaction $(X_i X_j + Y_i Y_j + Z_i Z_j)$, the desired property of off-diagonal matrix elements in the $Z$-basis can be obtained by a local unitary transformation~\cite{Marshall1955Antiferro}, $(Z \otimes I)\vec{\sigma}_i\cdot \vec{\sigma}_j (Z\otimes I) = (Z_i Z_j -X_i X_j - Y_i Y_j)$, but applying this transformation to successfully render the entire model stoquastic requires that the underlying interactions are \emph{bipartite} in the graph-theoretic sense (i.e., vertices are two-colorable). 

\paragraph{Hamiltonian Complexity.}
Computing the ground state energy of an AFM Heisenberg model on general graphs up to an inverse polynomial precision is known to be {\sf QMA}-complete, and is essentially regarded as the quantum analog of the classical {\sf MaxCut} problem~\cite{gharibian2019almost,anshu2020beyond}, thus called {\sf QuantumMaxCut} in the approximation algorithm community. 
In contrast, the same problem for {\it bipartite} graphs is contained in the more limited complexity class {\sf StoqMA} \cite{cubitt2016complexity} because it is stoquastic as described above. 
While classical {\sf MaxCut} is trivial to solve on bipartite graphs, the complexity of {\sf QuantumMaxCut} for bipartite graphs is unknown and is explicitly raised as an open problem in \cite{gharibian7faces}. Currenlty, the only proven fact is the containment in {\sf StoqMA} \cite{piddock2015complexity}, so the possibility is wide open from containment in {\sf P} to being {\sf StoqMA}-complete. 
Classical ground state approximations to the {\sf QuantumMaxCut} problem on general graphs have been obtained in special cases by semidefinite programming (SDP) heierarchies, generalizing the Goemans-Williamson algorithm for classical {\sf MaxCut} \cite{king2022improved,takahashi20232,watts2024relaxations,eunouojas2024}.  These SDP methods  are not explicitly affected by the sign problem, but bipartite graphs (including the star-like graphs we consider) have been an important source of tractable examples for these algorithms.

 QMC and related methods have been used for 
 proving polynomial-time simulability of a variety of stoquastic systems without {\sf NP}-hardness obstructions, including 1D spin systems at constant temperature~\cite{crosson2021rapid}, ferromagnetic Ising models on general interaction graphs~\cite{bravyi2014monte}, high-temperature quantum Ising spin glasses on bounded-degree graphs~\cite{crosson2020classical}, as well as a landmark result by Bravyi and Gosset establishing ground state and thermal state simulation for XY ferromagnets on arbitrary interaction graphs~\cite{bravyi2017polynomial}.  
 We note that this latter result includes $XX$ and $YY$ interactions as long as they are stoquastic (thus somewhat extending the notion of ``ferromagnetic'') and also local $Z$ fields, but not $ZZ$ interactions.  
 Another related result is a deterministic quasipolynomial-time simulation of XXZ models with tunable $XX,YY,$ and $ZZ$ interactions as long as the $ZZ$ term dominates the other two in a ferromagnetic way ~\cite{harrow2020classical}.

Based on the empirical evidence provided by decades of computational studies \cite{sandvik2010computational}, 
 it is fairly plausible that QMC algorithms can approximate ground states of bipartite {\sf QuantumMaxCut} models on a wide range of interaction graphs in classical polynomial time, but there has been no rigorous proof for efficient convergence of QMC for any bipartite AFM models to date.  Our work thus provides a first step along a reasonable approach to putting bipartite cases of {\sf QuantumMaxCut} in {\sf BPP}. 
 
\paragraph{Markov chains and Mixing Times.}  To obtain a polynomial-time simulation by QMC, the \emph{mixing time} of Markov chain that drives the QMC method must scale polynomially with the system size~\cite{levin2017markov}.   The mixing time of a Markov chain characterizes the number of classical updates which suffice to produce each approximate sample from the target distribution - a Markov chain which equilibrates in a time that scales logarithmically in the size of its state space (i.e. polynomially in the number of spins or particles for a quantum system) is called {\it rapidly mixing}.  

A general technique for analyzing the mixing time is to bound it in terms of reciprocal of the spectral gap of the Markov chain transition matrix.  To bound the spectral gap of the QMC dynamics for $O(1)$-bipartite graphs, we use the method of canonical paths due to Jerrum and Sinclair~\cite{jerrum1988conductance, sinclair1992improved}. 
In this method the Markov chain state space is viewed as a combinatorial graph, with vertices representing states of the Markov chain and edges representing transitions between states. 
The task of showing rapid-mixing is then essentially to show that this combinatorial graph is an expander. 
The canonical path method does this by showing the existence of a routing of the probability flow through the state space graph that delivers the necessary stationary probability to every vertex without overloading any particular edges.  
We apply the method of canonical paths in a relatively straightforward way (inspired by the textbook example of ``left-to-right bit fixing paths'') - most of our technical effort is devoted to developing the QMC representation that is simple enough to analyze yet close enough to the practical method, and proving geometric and topological facts about the state space so that these straightforward canonical paths suffice to show rapid mixing.   The success of this approach highlights the opportunity to apply more sophisticated Markov chain analysis techniques to the interdisciplinary open problem of simulating bipartite AFM by rigorously efficient QMC.  

\begin{figure}[t]
  \centering
  \includegraphics[scale=0.7]{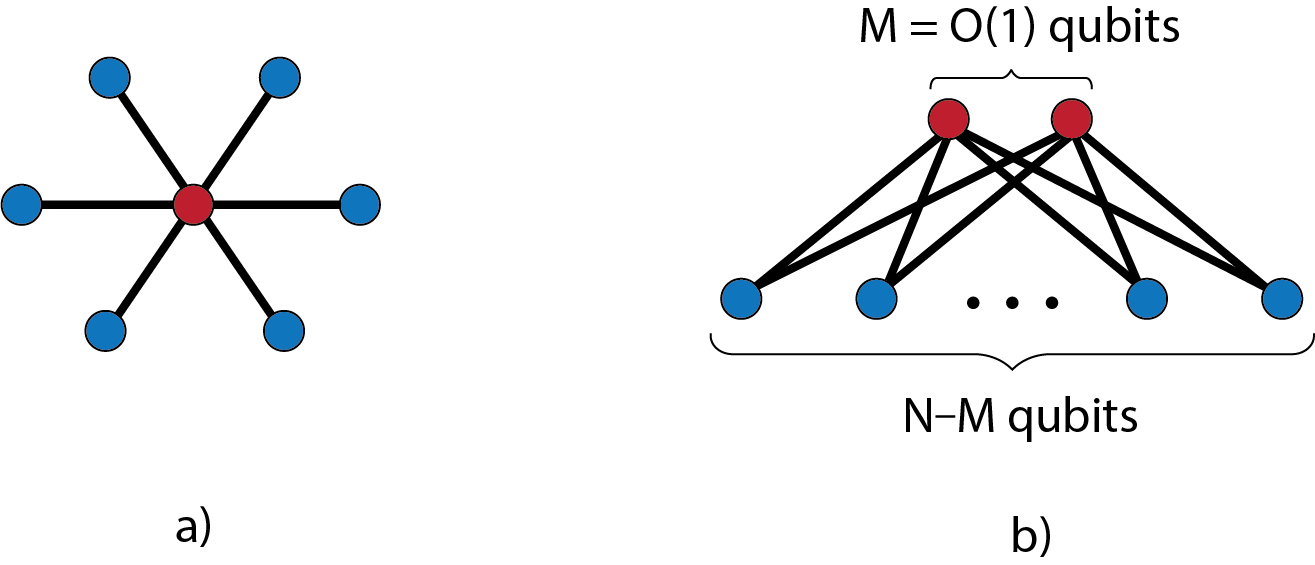}
  \caption{a) The ``Heisenberg star'', or star graph with $N-1$ qubits interacting antiferromagnetically with a single central qubit, is the archetypical example of the bipartite AFM we analyze. b) More generally we allow a constant number of qubits $M = O(1)$ to interact antiferromagnetically with $N - M$ qubits.  Our analysis also extends to include staggered local fields and ferromagnetic interactions within each bipartite component.}\label{fig:mainmodel}

  \label{fig:figure1}
\end{figure}
\paragraph{$O(1)$-bipartite (star-like) Heisenberg Models.}
Throughout this work we consider Heisenberg Hamiltonians defined on a bipartite interaction graph $\mathcal{A} \cup \mathcal{B}$ with $N = |\mathcal{A}|+|\mathcal{B}|$ qubits, in which one component has a fixed size $|\mathcal{A}|= O(1)$ and the other component has size $|\mathcal{B}| = O(N)$.   For any pair of qubits define the 2-local projection operator
\begin{equation}h_{ij} =  (\openone-X_iX_j-Y_iY_j-Z_iZ_j)/4.
\end{equation}
Then the Heisenberg AFM Hamiltonian with weights $w_{ij} \geq 0$ is
\begin{equation}\label{eq:HAFM}
H_{\textrm{AFM}} =- \sum_{i \in \mathcal{A},j\in \mathcal B} w_{ij} h_{ij}, 
\end{equation}
which we will explain in more detail in \cref{sec:notations}. 
Although the AFM terms are the primary interest in this work, our analysis also extends (without much additional complication) to including ferromagnetic terms acting within $\mathcal{A}$ or within $\mathcal{B}$, and 1-local $X$ terms with their signs staggered to enhance the antiferromagnetic order.  

\begin{align} 
H_{\textrm{FM}} &=- \sum_{(i,j) \in \mathcal{A} \textrm{ or } (i,j) \in \mathcal{B}}J_{ij}(\openone - h_{ij})\\
H_{X} &= \frac{1}{2}\sum_{i \in A} g_i(\openone-X_i)-\frac{1}{2}\sum_{j \in B} g_j (\openone-X_j) 
\end{align}
where $ J_{ij}, g_i$ are positive weights.  The most general Hamiltonian we consider is 
\begin{equation}
    H = H_\textrm{AFM} + H_\textrm{FM} + H_\textrm{X}\label{eq:fullHamiltonian}
\end{equation}
for $O(1)$-bipartite graphs, as indicated in \Cref{fig:mainmodel}.  \textbf{Our main result is a proof of $\textrm{poly}(N)$ time rapid mixing for the QMC Markov chain dynamics of \Cref{sec:introduceQMC} applied to these models, which allows for approximating the expectation of local observables and approximately sampling from the ground state in the $Z$ basis. }
  Since our results pertain to a QMC method that estimates the ground state energy of \eqref{eq:fullHamiltonian}, we emphasize that the ground state energy of these $O(1)$-bipartite models can always be found in $\textrm{poly}(N)$ time by the \textbf{Lieb-Mattis theorem} (see \Cref{app:LiebMattisTheorem}). Our motivation here is not just to find \emph{some} polynomial time algorithm for approximating these ground states, but to show that a popular and practical QMC algorithm also yields a rigorously efficient classical approximation algorithm for estimating ground state observables for these systems. The fact that the QMC algorithm is operated in a general way and does not make use of the Lieb-Mattis theorem to reduce the dimensionality brings hope for proving efficiency for more general cases in future work. 
Another perspective is that our main result lower bounds the spectral gap of the dynamical generator of a kinetic loop model that has not been previously analyzed. This loop model equilibrium distribution simulates the Heisenberg AFM ground state and thus yields an efficient algorithm for the problem of approximating the ground energy.

\paragraph{Loop representation of QMC for Heisenberg models.}  Following the discussion of stoquastic complexity, a central goal of computational QMC studies is to design Markov chains for which the state space, stationary distribution, and update rules lead to rapid mixing.   Despite progress in the rigorous analysis of QMC methods, bipartite Heisenberg models with AFM terms have resisted such analysis to date.  A high-level reason for this is that the form of the Heisenberg interaction leads to QMC configurations with many combinatorial constraints when using the most basic type of QMC \cite{Suzuki1977MonteCarlo}. 
Such constraints not only brings difficulties for the rigorous analysis of the Markov chain mixing time, but also results in impractical long mixing times.

To resolve this issue, modern QMC methods used in computational studies of Heisenberg models involve non-local cluster updates (analogous to Swendson-Wang updates \cite{swendsenwang} for classical Ising models, instead of single-site Glauber dynamics) achieving practically fast mixing time. 
Our approach is based on a state-of-the-art QMC used in computational condensed matter physics, called the stochastic series expansion (SSE) method \cite{Sandvik1991QMCSpinSys}.  This method is especially optimized for Heisenberg interactions, where it has been applied in thermal and ground state simulations with over $10^6$ qubits~\cite{takahashi2024so5}. 
In the simplified version of this method which we introduce and analyze, the stationary distribution over configurations of clusters (which are conventionally called ``loops'' even when they can be open strings) takes a simplified purely geometric form, proportional to $2^{\# \textrm{ of loops}}$.
Only by focusing on such practically fast QMC methods and appropriately simplifying them was it possible to rigorously prove rapid-mixing as we present here.

\begin{figure}[h]
  \centering
  \includegraphics[scale=0.2]{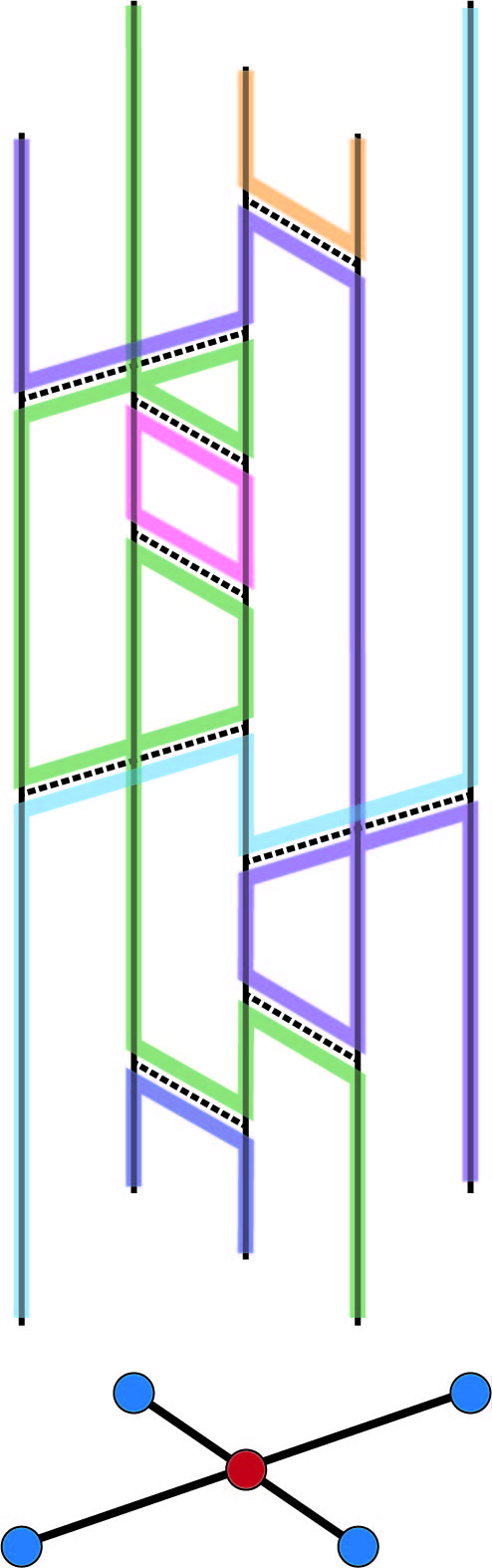}
  \caption{An example configuration from the SPE Markov chain  we analyze. The loop configuration (top) is an example of a configuration with $2B=8$ operators and 6 line segments (``loops'') for a Heisenberg model on a 4-arm star graph (bottom).}
  \label{fig:introSSE}
\end{figure}

The QMC we analyze closely follows the standard SSE algorithm, except for one difference that is significant for the analysis: instead of using the Taylor series for $\textrm{tr} [e^{-\beta H}]$ as is done in standard SSE, we use a fixed large positive integer power of the Hamiltonian $(-H)^L$  applied to the trial state $|+^N\rangle$ to project onto the low energy subspace. 
As we discuss later on, although physical intuition may indicate that this is not much of a difference, this in fact turns out to be crucial for our proof to work because of topological reasons. 
The idea of applying $(-H)^L$ to obtain the ground state has been done in many QMC methods  \cite{Sandvik_2010LoopValence}, but in our QMC, we use the $z$-basis unlike \cite{Sandvik_2010LoopValence} and then switch to the loop representation. 
We term this modified version of SSE as \emph{Stochastic Power Expansion} (SPE) QMC for convenience, as both the open boundary condition and the loop representation are essential for our analysis.

 \Cref{fig:introSSE} illustrates a typical SPE (SSE) configuration for the Heisenberg AFM on the star graph.  We briefly overview the method here; see \Cref{sec:introduceQMC} for a full derivation and explanation.
\begin{itemize}
\item For a general bipartite AFM, the SSE representation augments the interaction graph $\mathcal{G}$ with an ``imaginary time direction'' $\{1,...,2B\}$ and considers configurations consisting of an assignment of a ``bridge'' (corresponding to an edge in $\mathcal G$ i.e. a local AFM term, dotted in the figure) to each ``(imaginary) time slice'' $j \in \{1,...,2B\}$.   

\item One samples from this set of bridge configurations by a Markov chain that switches a randomly chosen bridge at each step. The Markov chain satisfies detailed balance and converges to the stationary distribution where probabilities are proportional to $2^{\# \textrm{ of loops}}$. The loops (which can be closed or open strings, shown in distinct colors for visual clarity in \Cref{fig:introSSE}) are defined via the bridge configuration. 

\item The loop structure of a configuration will in general undergo non-local changes when a single bridge is changed, which complicates the analysis but is essential for the efficiency of the method in practice.  
\end{itemize}

This paper is structured as follows. 
In \cref{sec:notations}, we introduce the notations as well as all possible forms of the Hamiltonians that we consider in this work. 
In \cref{sec:introduceQMC}, we derive the configuration space and the probability distribution that our QMC method will be sampling from. This derivation shows how the seemingly nonquantum yet topological configurations that we sample are related to the Heisenberg model. 
\Cref{sec:mixingtimemachinery} is devoted to introducing the canonical path method, the mathematical technique we use for proving rapid mixing. 
We also go through the basics of Markov Chain Monte Carlo methods in general, and also precisely define the QMC sampling algorithm based on the framework. 
We then move on to the main proof in \cref{sec:mainproof}, where we cover all cases that we can prove rapid-mixing. 
Finally, in \cref{sec:discussion}, we discuss the implications of our result as well as relations with various concepts in condensed matter physics and also comment on open problems. 

We advice readers who are only interested in the core idea of the algorithm to see 
the simplest configuration space we consider in \cref{eq:SSE-loop-state-space} and the probability distribution over it that is explained thereafter. 
This should be enough to understand the QMC procedure explained in \cref{sec:actualQMCprocedure}, and also understand our proof in the simplest case which is explained in the first few subsections of \cref{sec:mainproof}.

\section{Notations and Definitions}\label{sec:notations} 
The (antiferromagnetic, weighted) quantum Heisenberg model associated with the interaction graph $\mathcal{G}=(\mathcal{V},\mathcal{E})$ is an $N = |\mathcal{V}|$ qubit (spin-$1/2$) Hamiltonian with 2-local interactions. 
Let us denote the Pauli operators by $X$, $Y$, $Z$, and adopt the usual abbreviation $X_i$ for the operator acting only on the $i$-th qubit nontrivially, i.e. $X_i\coloneq \openone\otimes\openone\ldots\otimes X\otimes\ldots\otimes\openone$, where $\openone$ is the identity operator on the 2-dimensional Hilbert space. $Y_i$, $Z_i$ are similarly defined.

The eigenstates of the the Pauli $Z$ operator are denoted as 
$Z\upket = \upket$ and $Z\downket = -\downket$. The eigenstates of the Pauli $X$ operator will be denoted as $\ket{+} =\frac{1}{\sqrt{2}}( \upket + \downket)$ and $\ket{-} =\frac{1}{\sqrt{2}}( \ket{\hspace{-0.8mm}\uparrow} - \ket{\hspace{-0.8mm}\downarrow})$ while adopting the standard tensor product conventions of $\upket\otimes\upket \eqqcolon \ket{\hspace{-1.4mm}\uparrow\uparrow}$ and $\ket{\hspace{-0.8mm}\uparrow}^{\otimes N} \eqqcolon \ket{\hspace{-0.8mm}\uparrow^N}$.

The antiferromagnetic Heisenberg Hamiltonian we consider will be defined as 
\begin{equation}\label{eq:HeisenbergHGeneral}
    H=-\sum_{(i,j)\in \mathcal{E}} w_{ij} h_{ij}, 
\end{equation}
where the summation $(i,j)$ runs through all the edges $\mathcal{E}$ of a graph $\mathcal{G}=(\mathcal{V},\mathcal{E})$. 
Here, the indices $i$ and $j$ correspond to spins (thus $|\mathcal{V}|=N$), and $h_{ij}$ is the Heisenberg interaction term $h_{ij}\coloneq (\openone-X_iX_j-Y_iY_j-Z_iZ_j)/4$ between those two qubits. 
This convention of the term makes the $h_{ij}$ to be the projector on to the the singlet state $\left(|\updown\rangle - |\downup\rangle\right)/\sqrt{2}$. 
The summation in the Hamiltonian $H$ has $M=|\mathcal{E}|$ terms where $w_{ij}$ are positive weights for the individual terms. 

In this work, we focus on the case when $\mathcal{G}$ is bipartite, i.e., $\mathcal{V}=\mathcal{A} \cup \mathcal{B}$ and edges only connect spins between $\mathcal{A}$ and $\mathcal{B}$. Thus, $i,j\in \mathcal{A} \rightarrow (i,j)\notin \mathcal{E}$, and same for $\mathcal{B}$. 
If \eqref{eq:HeisenbergHGeneral} is defined over such a bipartite graph, we can assume that the local unitary transformation 
\begin{align}\label{eq:stoquastic_bipartite-rotation}
\tilde{h}_{ij} \coloneq (Z_i \otimes \openone_j) ~h_{ij}~ (Z_i\otimes \openone_j)
\end{align}
has been carried out on each term resulting in the stoquastic version of $H$: 
\begin{align}\label{eq:StoqMaxCutHam}
    \tilde{H} 
    = - \sum ~ w_{ij}~ \underbrace{\frac{\openone+X_iX_j+Y_iY_j-Z_iZ_j}{4}}_{\tilde h_{ij}}.
\end{align}
Note that the bipartiteness is crucial here, since that is the only case when 
the total unitary transformation $U\coloneq\prod_{i\in \mathcal A}Z_i$ is well-defined to obtain 
the above $\tilde{H}=U H U^{\dagger}$.
The local terms of $\tilde H$ can then be decomposed as:
\begin{align}\label{eq:localhStoqaustic}
    \tilde h_{ij} =  \frac{1}{2}\left(I_{ij} +  S_{ij}\right),
\end{align}
where 
\begin{align}\label{eq:IdentNeel}
    I_{ij} = \ketbra{\hspace{-0.8mm}\uparrow\downarrow}{\uparrow\downarrow\hspace{-0.8mm}} + \ketbra{\hspace{-0.8mm}\downarrow\uparrow}{\downarrow\uparrow\hspace{-0.8mm}}
\end{align}
and
\begin{align}\label{eq:SwapNeel}
     S_{ij} = \ketbra{\hspace{-0.8mm}\uparrow\downarrow}{\downarrow\uparrow\hspace{-0.8mm}} + \ketbra{\hspace{-0.8mm}\downarrow\uparrow}{\uparrow\downarrow\hspace{-0.8mm}} .
\end{align}
It is useful to think of $ I_{ij}$ and $S_{ij}$ as the identity and swap operators respectively on the antialigned subspace: $\text{span}\left(\left\{   \ket{\updown},\ket{\downup}   \right\}\right)$.
Note that the basis rotation by $Z$ had the effect of changing the sign accompanying $S_{ij}$ thanks to the bipartiteness. 

We will denote the ground state energy (minimum eigenvalue) as $E_0$ without explicitly writing the dependency to the Hamiltonian, since in all cases it will be clear from the context.

In later sections, we will consider ferromagnetic interactions too, but {\it only those that connects among spins on the same side of the bipartition}. Note these terms are already stoquastic and will be unaffected by the transformation \eqref{eq:stoquastic_bipartite-rotation} that changes the anti-ferromagnetic terms into their stoquastic counterparts.  
In this case, we use a slightly different convention for the ferromagnetic terms as 
\begin{equation}\label{eq:FMterms1}
    -\sum_{(i,j)\in\mathcal{E}_{\mathrm{F}}} v_{ij} (\openone-h_{ij})  
\end{equation}
with positive $v_{ij}>0$ connecting $i$ and $j$ on the same side, denoted by $\mathcal{E}_\mathrm{F}$. 
In this case, we can decompose the term as 
\begin{equation}\label{eq:FMterms2}
    \openone-h_{ij} = \frac{1}{2}\left( I_{ij}^{\mathrm{F}} + S_{ij}\right), 
\end{equation}
which $S_{ij}$ is the same swap operator as in Eq. \eqref{eq:SwapNeel}, but $I^{\mathrm{F}}$ is now the identity operator on the ferromagnetic subspace:
\begin{align}\label{eq:IdentFM}
    I_{ij}^{\mathrm{F}} = \ketbra{\hspace{-0.8mm}\uparrow\uparrow}{\uparrow\uparrow\hspace{-0.8mm}} + \ketbra{\hspace{-0.8mm}\downarrow\downarrow}{\downarrow\downarrow\hspace{-0.8mm}} . 
\end{align}
Note that these ferromagnetic terms will not be affected by the $Z$ conjugation, since either two $Z$ operators or none would be applied to $h_{ij}$.

We will also consider adding 1-local transverse fields 
\begin{equation}\label{eq:transversefield}
    -\sum_{i\in\mathcal V} g_i (\openone\pm X_i)  = 
       -\sum_{i\in\mathcal A} g_i (\openone- X_i)     -\sum_{i\in\mathcal B} g_i (\openone + X_i) 
\end{equation}
where $g_i\geq 0$ and the sign depends on whether $i\in\mathcal A$ or $i\in\mathcal B$, so that after applying the $Z$ conjugation, there is no sign problem.

\section{Modified Stochastic Series Expansion Quantum Monte Carlo}\label{sec:introduceQMC}

In this section we introduce the QMC method we employ to prove rapid-mixing, which we dub the Stochastic Power Expansion (SPE) QMC. 
This can be seen as a modified version of the stochastic series expansion (SSE) Monte Carlo method \cite{Sandvik1991QMCSpinSys,Sandvik1992GenHands,Sandvik_1999OpLoop}. 
We blend desirable attributes of the standard SSE for finite temperature systems and the projector QMC in the so called ``valence-bond'' basis \cite{sandvik2007monte} designed for ground states while also incorporating ideas from random loop models \cite{ueltschi2012quantum,aizenman1994geometric,aizenman2020dimerization,goldschmidt2011quantum,Bj_rnberg_2019,Nahum_2013,motamarri2023loop}, all used and studied in condensed matter physics. 
Our loop representation generalizes the Temperley-Lieb algebra \cite{temperley2004relations,kauffman1987state} from 1-dimensional quantum spin chains to more general graphs, providing a framework for complexity analysis of the associated stochastic dynamics, which had not been explored previously. 
Recently a similar method was proposed and used to study a generalized SU($n$) version of the Heisenberg model on a bipartite lattice \cite{Desai2021Resummation}. As we show here, the SPE framework not only allows for  a natural SU($n$) extension, but also happens to yeild the simplest Markov Chain Monte Carlo to analyze the mixing time for Heisenberg systems. 

We shall note that the most nontrivial part in our construction of this SPE QMC design is the derivation of the configuration space that we explain in the following. Once we define the probability distribution that we would like to sample from, the actual procedure to sample those configurations with a Markov chain is fairly simple, and is concisely explained in \cref{sec:actualQMCprocedure}. 
How one could use such efficient sampling to compute physical observables including the ground state energy is a well known procedure in the computational physics QMC studies, and is explained in \cref{sec:ObsAndEnergyEst}. 

\subsection{Basic transformations for QMC}
We will first explain how to derive our QMC for the 
uniformly antiferromagnetic bipartite Hamiltonian 
\begin{equation}\label{eq:SimpleHam}
    H_0 = - \sum_{(i,j)\in\mathcal E} h_{ij}, 
\end{equation}
and then explain how to generalize it to 
the most generic form of the Hamiltonian which we consider here proving rapid mixing: 
\begin{align}\label{eq:GeneralHam}
    H =  &- \sum_{(i,j)\in\mathcal E} w_{ij}h_{ij}
    - \sum_{(k,l)\in \mathcal E_F} v_{kl}(\mathbbm 1 - h_{kl})\nonumber \\
    &- \sum_{m\in \mathcal V} g_m (\mathbbm 1 \pm X_m).
\end{align}

 The QMC we consider is based on applying $(-H_0)^B$ to the state $\ket{+^N}\coloneqq |+\rangle^{\otimes N}$ for some large positive integer power $B$, to form the state 
 \begin{equation}
 |M_B\rangle = \frac{(-H_0)^B|+^N\rangle}{\sqrt{ \langle +^N |(-H_0)^{2B}|+^N \rangle}}
 \end{equation}
 Since the ground state of $H_0$ has nonzero overlap with the $\ket{+^N}$ state, $|M_B\rangle$ will converge to the ground state $\ketGS$ exponentially fast as a function of $B$.  In order for this power method to truly converge to the ground state, $B$ would need to scale with the inverse spectral gap of $H$.  However, for the goal of estimating the ground energy within additivie precision $\epsilon$, it suffices to  project onto a subspace of states within energy $\epsilon$ above the ground state, and therefore $B$ can scale with the precision $\epsilon$ instead of the spectral gap.    Denote the projection operator onto the subspace of states with energy in the range $[E_0,E_0 + \epsilon]$ by $\Pi_{[E_0, E_0 + \epsilon]}$.   In \Cref{sec:scalingOfB} we prove that $\langle M_B |\Pi_{[E_0, E_0 + \epsilon]}|M_B\rangle > 1 -\delta$ when
$$B \geq  \frac{\|H\|}{2 \epsilon} (N + \ln (\delta w)^{-1})$$
where $w > 2^{-N}$ is the overlap of the initial trial state ($|+^N\rangle$ in this case) with the low energy subspace $\Pi_{[E_0,E_0 + \epsilon]}$.   This implies that a polynomial time algorithm with $B = \textrm{poly}(N)$ can take the leakage out of the low energy subspace $\delta$ to be exponentially small and the precision of energy estimate $\epsilon$ to be inverse polynomially small.   Therefore in the following we assume $B$ is sufficiently large that an observable $O$ can be estimated to desired precision via 
\begin{equation}\label{eq:expectation}
\langle O\rangle := \braGS O \ketGS 
\cong 
\frac{\bra{+^N}(-H_0)^B O (-H_0)^B\ket{+^N}}{\|\bra{+^N} (-H_0)^{2B} \ket{+^N}\|}. 
\end{equation}

As we show in \Cref{sec:ObsAndEnergyEst}, this estimation problem could be reduced to the sampling problem of the denominator 
\begin{align}\label{eq:veryfirstZ}
\mathcal Z &:= \bra{+^N}(-H_0)^{2B} \ket{+^N}\\ 
= &\frac{1}{2^N}\sum_{\sigma_L, \sigma_R}\langle \sigma_L | \biggl(\sum_{(i,j)\in \mathcal E} h_{ij}\biggr)^{2B}\hspace{-2mm} |\sigma_R\rangle \\ 
= & \frac{1}{2^{N+2B}}\sum_{\sigma_L, \sigma_R}\langle \sigma_L | \biggl(\sum_{(i,j)\in \mathcal E} {I}_{ij} + {S}_{ij}\biggr)^{2B}\hspace{-2mm} |\sigma_R\rangle, \label{eq:Z}
\end{align}
where $\sigma_L$ and $\sigma_R$ are $Z$-basis states and ${I}_{ij}$ and ${S}_{ij}$ are as defined in $\eqref{eq:IdentNeel}$ and \eqref{eq:SwapNeel} respectively. Note that moving forward we will neglect the constant $2^{-(N+2B)}$ factor as it will cancel out when estimating observable expectations as in \eqref{eq:expectation}.
By expanding the $2B^{\text{th}}$ power, we get a sequence of $I$s and $S$s, each acting on some edge $(i,j)\in \mathcal E$.
Let us denote with $\mathfrak{S}$ the set of all such sequences,
and $W(t)$ to be the $t$-th element (which is an operator) in a sequence $W\in\mathfrak{S}$. 
We then are able to further rewrite $\mathcal Z$ as 
\begin{equation}\label{eq:SSEZ}
\mathcal Z =\sum_{\sigma_L, \sigma_R, W\in\mathfrak{S}}\langle \sigma_L | \prod_{t=1}^{2B} W(t) |\sigma_R\rangle, 
\end{equation}
which is fairly close to to the typical finite temperature SSE, but with an open boundary condition. 
Upon inspecting Eq.~\eqref{eq:SSEZ}, one can notice that each term in the summation can only become either 0 or 1 depending on its consistency; if the two boundary conditions $\sigma_L$ and $\sigma_R$ are consistent with all the intermediate $S$ and $I$ operators in the sequence $W$, 
the term evaluates to $1$, otherwise, it is $0$. Also notice that once $\sigma_R$ and $W$ is fixed, there is only one $\sigma_L$ that is consistent with them (if $\sigma_R$ and $W$ themselves are consistent). 
Therefore, we can even further simplify 
\begin{align}\label{eq:ConsistentCount}
Z= \sum_{\sigma_L,\sigma_R, W\in\mathfrak{S}} \mathbb{I}_{[\sigma_L,\sigma_R, W \mathrm{~are~ consistent}]} 
\end{align}
which is simply the number of consistent choices of $\sigma_R, W\in\mathfrak{S}$, which from now on we will just refer to as ``configurations.''

\subsection{Loop Representation}\label{sec:looprep}

Now, we switch to another representation that has two effects. 
First, the new representation naturally eliminates the ``hard constraints'' of consistencies present in \cref{eq:ConsistentCount}. 
Second, it also provides a purely geometric and topological way of viewing these configurations. 
Both of these points are closely related and contribute to the final simplicity of our analysis. 

In order to motivate the representation, we look at the matrix element of a consistent configuration $(\sigma_R, W)$, given by $\bra{\sigma_L}\prod_{t=1}^{2B}W(t)\ket{\sigma_R}$. If we insert $2B-1$ complete sets of $z$-basis states in between each of the $W(t)$ operators, we get the expansion:
\begin{align}\label{eq:ConsistentSeq}
    \bra{\sigma_L}W(1)\ketbra{\sigma_1}{\sigma_1}W(2)\ket{\sigma_2}
    \cdots\bra{\sigma_{2B-1}}W(2B)\ket{\sigma_R}
\end{align}
where the sums over the inserted states do not appear because for a given operator $W(t)$ and $z$ basis state $\sigma$, there is a unique $Z$-basis state $\sigma'$ such that $\bra{\sigma}W(t)\ket{\sigma'}$ is non-zero (1 actually).

\begin{figure}[h!]
  \centering
  \includegraphics[scale=0.5]{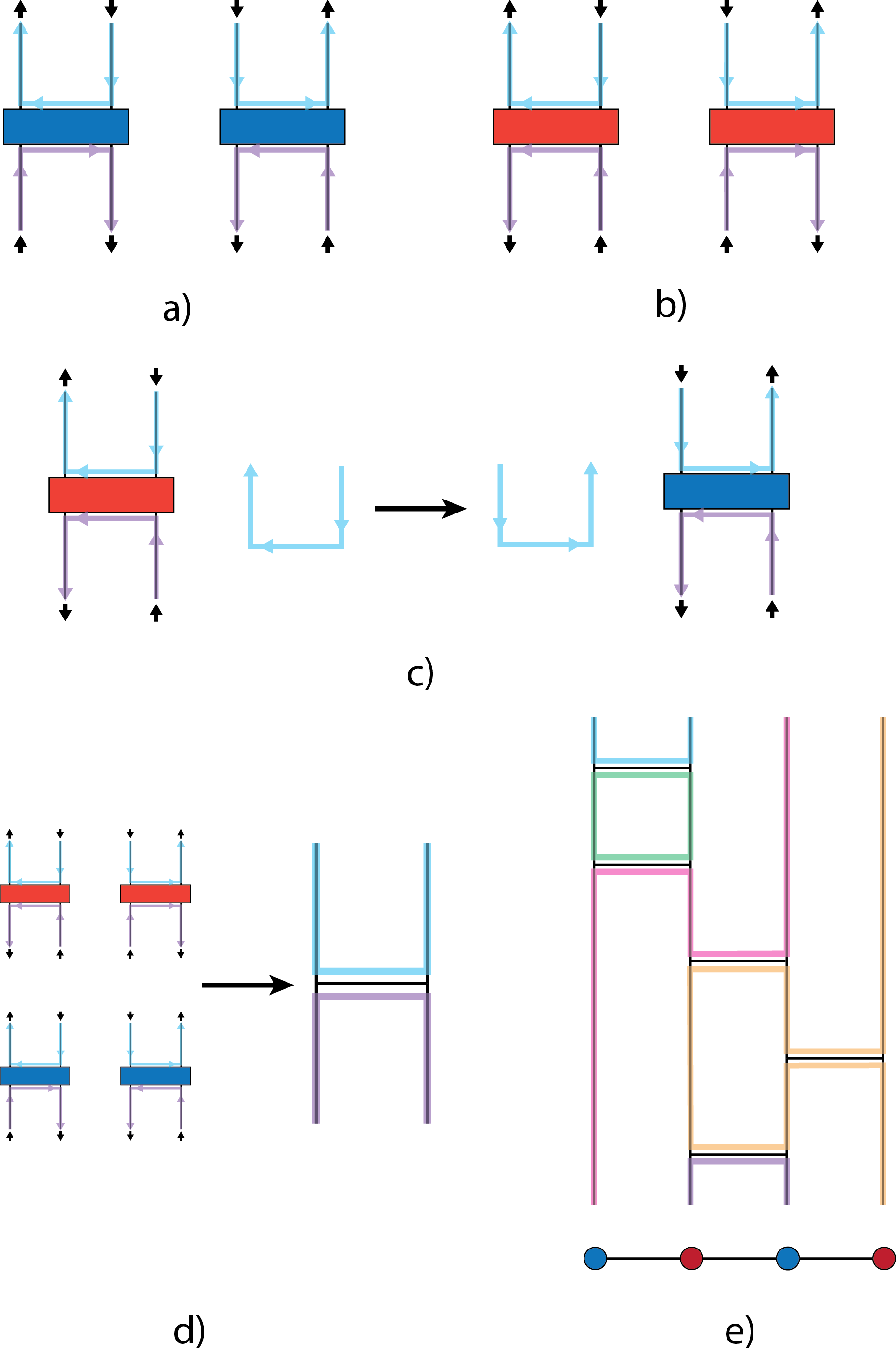}
  \caption{ a) The consistent spin configurations when an $I$ operator appears inserted in the inner product \eqref{eq:ConsistentCount}. The consistent configurations equivalently described by the blue and purple directed loops that make "u-turns" when they encounter an operator. b) The consistent spin configurations and associated directed loops for when an $S$ operator appears in \eqref{eq:ConsistentCount}. c) Consistent configurations can be switched between by reversing the direction of a directed loop, and then interchanging all of the operators that the loop touches from $I\leftrightarrow S$. Note that when the loops are stitched together as in e), if a loop touches both sides of an operator
  then that operator must always be $I$ and does not interchange when the direction of the loop is reversed. 
  d) Given that the total number of consistent configurations is given by $2$ raised to the power of the number of loops, we shift to an representation where all of the possible configurations in a) and b) are represented by un-directed loops and no distinction is made between the $I$ and $S$ operators. e) For a full operator string \eqref{eq:ConsistentSeq}, the local loops are joined together. The example here is for 4 qubits on a line, in the configuration $ x = ((1,2)(1,2),(2,3),(3,4),(2,3))$. Since there are $5$ loops in $x$ the steady state weight of this configuration is proportional to  $2^{L(x)} = 2^5$.  }
  \label{fig:antiFerroLoopExplanation}
\end{figure}

The key property that is used to define the loops is that whenever an operator $W(t)$ is placed on (say) edge $(i,j)\in \mathcal E$, the spins on the $i$ vertex and the $j$ vertex must be anti-aligned, regardless of whether the operator $W(t)$ is an $I$-type or an $S$-type (remember \cref{eq:localhStoqaustic}). This is also true for the spins after the operator is applied. 
In other words, the $z$-basis states $|\sigma_t\rangle$ and $\langle \sigma_{t-1} |$ must both have the spins anti-aligned. The $2^2=4$ different possibilities are depicted in  \Cref{fig:antiFerroLoopExplanation} (a) and (b). 
The two different types of operators are represented with different colored rectangles, and the $Z$-basis configurations are drawn below and above the operators, corresponding to before and after applying the operator $W(t)$. 

The observation here is that if we consider two ``line segments'' that make a U-turn around the operators below and above the operator (rectangle), each line segment always connects anti-aligned spin configurations in the two sublattices, and therefore flipping the $Z$-configuration {\it along an entire line segment} always results in another valid configuration. An example is shown in  \Cref{fig:antiFerroLoopExplanation} (c), where flipping along the blue line in a configuration in (b) results in a configuration in (a) after recoloring the rectangle (i.e., redifining the operator type); both configurations are valid and have the same weight. 
This allows for the construction of ``loops'' in the following manner. Start from any point on any line representing a particular spin $i$ in the diagrammatic description of $W$, and follow the line in the imaginary time direction until an operator $W(t)$ is encountered that connects vertex $i$ with another vertex $j$. Then make a U-turn after crossing the rectangle to the $j$ side, and keep moving in that direction. 
Since the construction of such paths is deterministic once $W$ is specified and never branches, disappears or back tracks, there are only two ways how this procedure can terminate: either the path comes back to the original starting point creating a loop, or has exactly two end points at $\langle\sigma_L|$ or $|\sigma_R\rangle$. 

Now, as shown in \Cref{fig:antiFerroLoopExplanation} (c), the set of consistent configurations of Eq. \eqref{eq:ConsistentCount} with an arbitrary operator acting on a fixed edge can be switched between by reversing the direction of the loops and updating the operator accordingly. This implies that the total number of consistent configurations for a fixed edge sequence is simply the $2^{\text{\# of loops}}$. The local loop rules are then stitched together as in e) for a configuration \cref{eq:ConsistentSeq}.

 We thus define the state space $\Omega$ of these configurations via:
 \begin{align}\label{eq:SSE-loop-state-space}
     \Omega = \{(e_1,...,e_{2B}):e_i\in \mathcal E\}.
 \end{align}
  For a given configuration $x\in\Omega$, we will define $L(x)$ as the function that returns the number of loops in configuration $x$. This allows us to re-express the (effective) partition function $\mathcal Z$ as:
\begin{align}\label{eq:loopz}
    \mathcal Z = \sum_{x\in\Omega}2^{L(x)}.
\end{align}

We show in \Cref{sec:ObsAndEnergyEst} that local Hamiltonian terms (and any other observables that can be incorporated into the loop representation) can be estimated using samples  $x$ drawn from the distribution:
\begin{align}\label{eq:piEquilibrium}
    \pi(x) = \frac{2^{L(x)}}{\mathcal Z}
\end{align}
and so samples from this distribution can be used to estimate the energy of $H$. 

In later explanations, it will occasionally be convenient to imagine that the individual loops can have either of the two $Z$-basis configurations associated with them. 
In such cases, we will refer to the loops with such $Z$-basis configurations attached to be {\it directed} loops (as in \Cref{fig:antiFerroLoopExplanation} a), b) and c)), as opposed to those that are {\it undirected} (\Cref{fig:antiFerroLoopExplanation} e))\footnote{This should not be confused with {\it directed loop updates} commonly used in SSE for anisotropic Heisenberg models \cite{Syljuaasen02Directed}.}. 

\subsection{Additional terms in the Hamiltonian} 

The previous explanation only considered cases where we have an unweighted bipartite graph $G$ that represents the interaction of the Heisenberg model. 
As we explained in \Cref{sec:notations}, in this work we also consider extensions of the Hamiltonian in the following three different ways: 
\begin{enumerate}
\item the Hamiltonian can have nonuniform weights $w:\mathcal{E}\rightarrow\mathbb{R}_{\geq 0}$ as long as the weights are positive and they satisfy the bipartiteness. 
\item the Hamiltonian can also include Heisenberg exchange interactions between spins on the same side of the bipartition, as long as they are ferromagnetic. 
\item We can also add single-body transverse field terms as explained in Eq. \eqref{eq:transversefield}. 
\end{enumerate}

\subsubsection{Extension with Weights}\label{sec:ExtWeights}
When we have weights $w_{ij}$ for each term $h_{ij}$ in the Heisenberg hamiltonian, that simply transfers to the weights in Eq. \eqref{eq:Z} as 
\begin{equation}
    \mathcal Z =  \frac{1}{2^{N+2B}}\sum_{\sigma_L, \sigma_R}\langle \sigma_L |  \biggl(\sum_{(i,j)\in \mathcal E}w_{ij}\left( {I}_{ij} + {S}_{ij}\right)\biggr)^{2B}\hspace{-2mm} |\sigma_R\rangle.
\end{equation}
Since the edge weights $w_{ij}$ are independent of whether there is an $I$ or an $S$ in the sequence, we can straightforwardly incorporate them into the loop representation. In order to do this, define a function $T$ from an edge $e = (i,j)\in \mathcal E$, to the real numbers via $T(e) = w_{ij}$. If for a given loop representation configuration $x = (e_1,...,e_k,...,e_{2B})$ we define the $k^{\text{th}}$ element via $x_k = e_k$ then the partition function with weights can be written as:
\begin{align}\label{eq:loopzwithweights}
    \mathcal Z = \sum_{x\in\Omega}2^{L(x)}\prod_{k=1}^{2B}T(x_k).
\end{align}
In this case the the steady state weight of a configuration $x$ becomes:
\begin{align}\label{eq:stationarywithweights}
\pi(x) = \frac{2^{L(x)}\prod_{k=1}^{2B}T(x_k)}{\mathcal Z}, 
\end{align}
which trivially includes the unweighted case where $T(x_i)\equiv 1$.

\subsubsection{Ferromagnetic Bonds}
The SPE method can also handle cases with ferromagnetic interactions as long as they only appear within the same sublattice. 
Here, for simplicity, let us first consider the case where all such ferromagnetic terms have uniform weights $v_{ij}\equiv 1$. 

\begin{figure}[h!]
  \centering
  \includegraphics[scale=0.08]{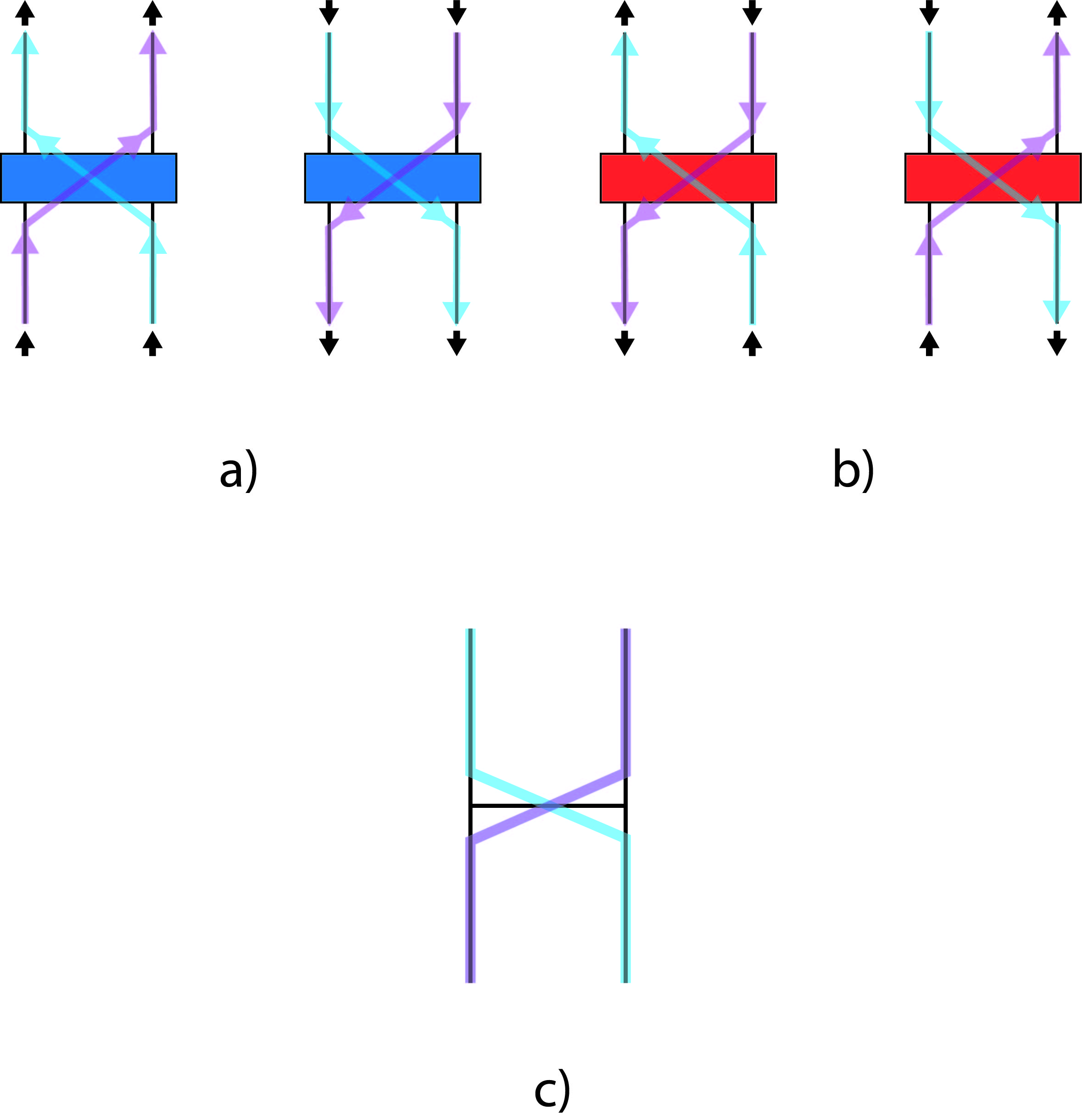}
  \caption{a) The possible spin configurations and associated directed loops for the $I^F$ operator coming from a ferromagnetic term.   b) The possible spin configurations and associated directed loops for the $S$ operator coming from a ferromagnetic term. c) The ferromagnetic term in the un-directed loop representation.}\label{fig:FerroTerms}
\end{figure}

By plugging in Eqs. \eqref{eq:FMterms1} and \eqref{eq:FMterms2} to Eq. \eqref{eq:Z}, we obtain Eq. \eqref{eq:loopz} with a simple modification for the definition of ``loops''. 
Namely, the the ferromagnetic bonds ``swap'' the two lines coming from the same imaginary time direction, as opposed to the antiferromagnetic bonds connecting two lines on the same same imaginary time direction. 
As depicted in \Cref{fig:FerroTerms}, this corresponds to the two possibilities of $I_{ij}^{\mathrm{F}}$ or $S_{ij}$, and is a standard method for incorporating ferromagnetic interactions in SSE \cite{Sandvik1991QMCSpinSys}. Here the state space $\Omega$ becomes the set of all tuples $(e_1,...,e_{2B})$ where now $e_i \in \mathcal E\cup \mathcal E_{\mathrm F}$.

Extending this to nonuniform weights (that satisfy the sign constraint $v_{ij}> 0$) is straightforward. We end up in the same expression as \cref{eq:loopzwithweights,eq:stationarywithweights}, just with the extended rule for loop construction/counting $L(x)$ as we explained above.

\subsubsection{Transverse Fields}\label{subsubsec:TF}
Our SPE also admits transverse field terms in the Hamiltonian, in this case, only if they are stoquastic. Again, we consider the case of uniform weights for explaining ($g_i, w_{ij},v_{ij}\equiv1$), and the extension to the weighted case is straightforward. 

We start from noting that the identity operator $\openone + X_i$ in \cref{eq:transversefield} is there for technical convenience with the loop representation. 
While adding a constant value to the Hamiltonian does not change any of the physics, it enables us to treat the transverse field as an endpoint of the loops (strips) in the loop representation. 

To see how this works, we must first name the identity that comes together with $X_i$ to be $I_i$ and distinguish them. 
With these terms available in $\mathfrak{S}$, each $W(t)$ in \cref{eq:SSEZ} can be an actual identity $I_i$, or a ``flipper'' $X_i$. 
The four different possibilities that can happen at these two operators are depicted in \cref{fig:XLoopExplanation} (a) and (b). Crucially, these all have the same weight because we added $\openone$ for each $X$ term. 

\begin{figure}[t]
  \centering
  \includegraphics[scale=0.33]{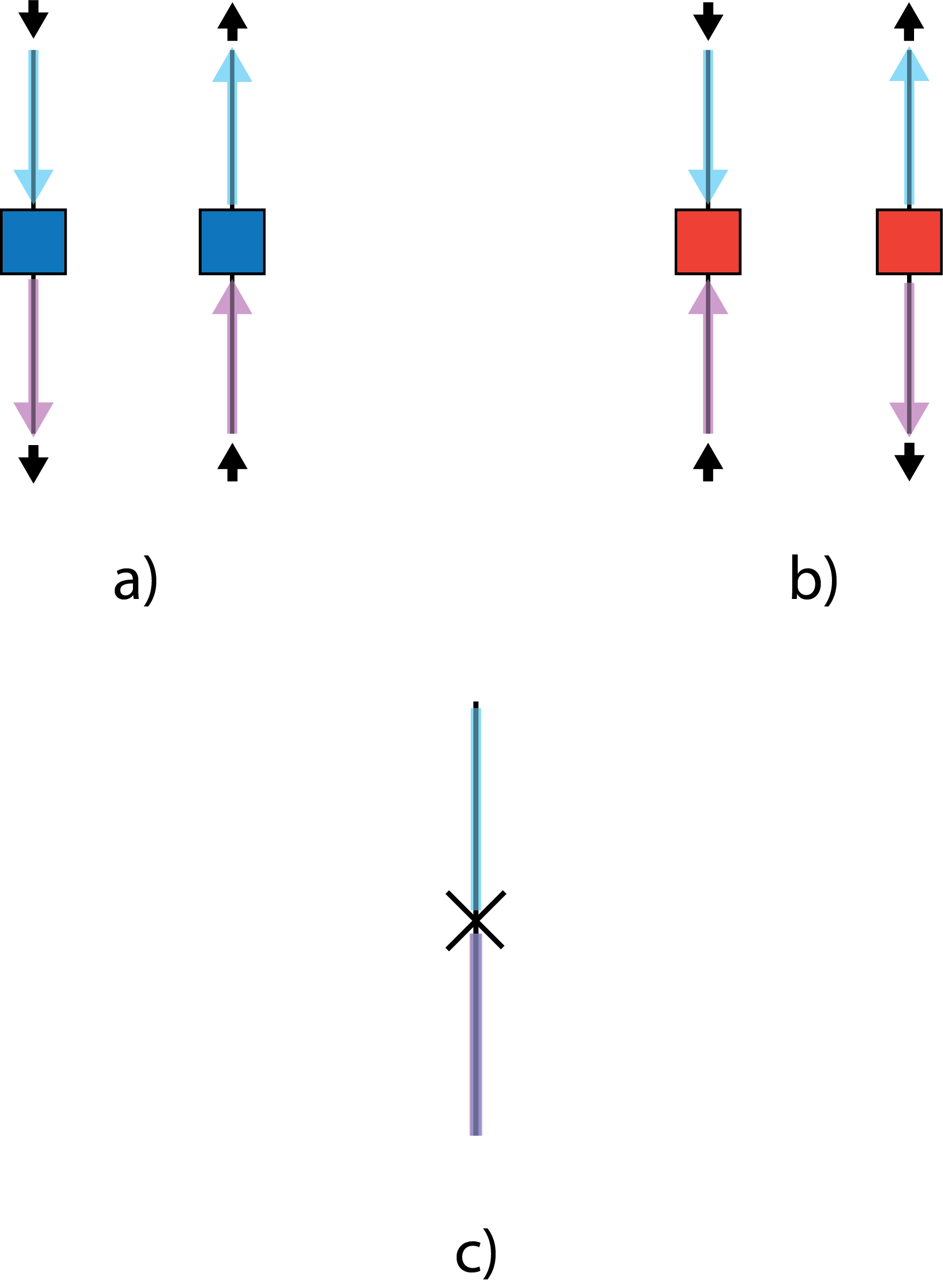}
  \caption{ a) The possible spin configurations and associated directed loops for the  $\mathbbm 1$ operator coming from a transverse field term.   b) The possible spin configurations and associated directed loops for the $X$ operator coming from a transverse field term. c) The transverse field term in the un-directed loop representation. }
  \label{fig:XLoopExplanation}
\end{figure}

Similarly to how we arrived to the loop representation in \cref{sec:looprep} and for the ferromagnetic case,  
we can interpret this situation as having a single type of ``vertex operator'' as in \cref{fig:XLoopExplanation} (c) that can be either $I_i$ or $X_i$ depending on the four different cases of what the upper and lower strings are in the $Z$-basis.
This means we again end up with the same expression Eq. \eqref{eq:loopz} but with an extended space of ``sequences'' $\Omega$. Here, the modified notion of loops $L(x)$ changes with one additional rule that they terminate at these transverse-originated operators. 
Now configurations $x\in\Omega$ can also have terms corresponding to a vertex (as long as a transverse field is applied to it), thus $x=(e_1, e_2, \ldots, e_{2B})$ where $e_i\in\mathcal{E}\cup\mathcal E_{\mathrm F}\cup \mathcal{V}$. 
The loops explained in \cref{fig:antiFerroLoopExplanation} will terminate at endpoints created by those vertex operators, as depicted in \cref{fig:GenericSPEconfig}. $L(x)$ counts the number of such line segments, either those that terminate at endpoints at the boundaries or with the vertex operators, and also proper loops without boundaries. 

In the case where we have nonuniform weights $g_i$, the weights simply includes the factors coming from them which leads to the same expression \cref{eq:loopzwithweights,eq:stationarywithweights} just with the modified way of counting $L(x)$.

\begin{figure}[t]
  \centering
  \includegraphics[scale=0.27]{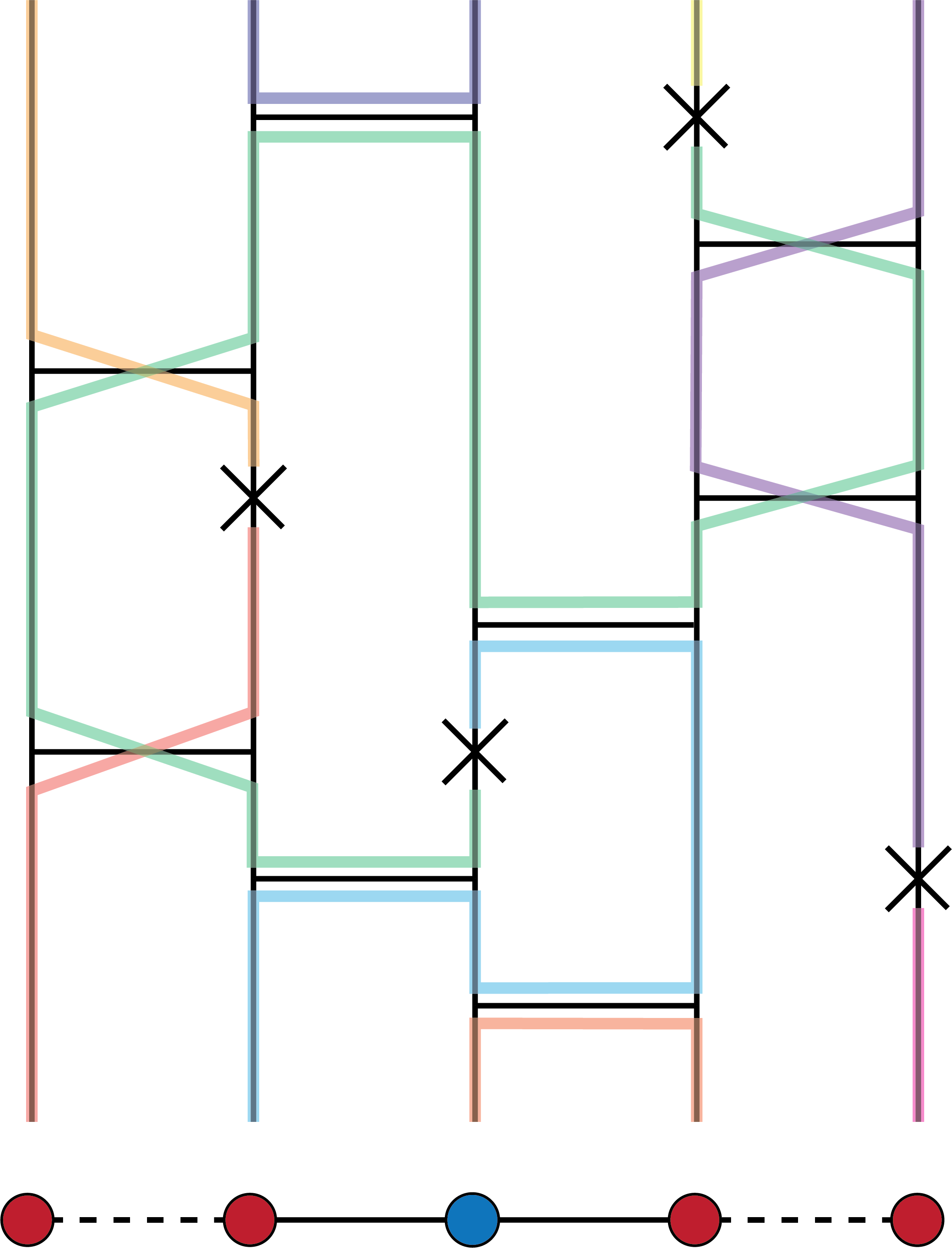}
  \vspace{3mm}
  \caption{An example configuration with included ferromagnetic and transverse field terms. The example here has $2B=12$ operators and $L(x)=9$ loops (line strips). The graph at the bottom shows the interaction of the Heisenberg model: Solid lines for AFM interactions and dotted lines for FM interactions. The vertex colors represent the bipartiteness of the graph.}
  \label{fig:GenericSPEconfig}
\end{figure}

\section{Mixing Time Machinery}\label{sec:mixingtimemachinery}

In this section we review the relevant aspects of Markov chain Monte Carlo (MCMC) and Markov chain mixing analysis. 
We then explain one specific MCMC that will be subject to our analysis that converges to the desired steady-state distribution \cref{eq:stationarywithweights}.

\subsection{Markov Chains and Mixing Times}\label{sec:mcmcbasics}
We consider Markov chains defined by a finite state space $\Omega$, a stochastic matrix $P$ describing transition probabilities between elements of $\Omega$, and a stationary distribution $\pi$ on $\Omega$ that satisfies $\sum_{x\in \Omega} \pi_{x} P_{xy}  = \pi_y$ for all $y\in \Omega$ or $\pi P = \pi$, where we treat the distribution $\pi$ as a horizontal vector.  The transition probabilities satisfy $P_{xy} \geq 0$ and $\sum_{y} P_{xy} = 1$ for all $x,y\in \Omega$.   The probability to go from $x$ to $y$ after $t$ steps of the chain is denoted $P^t(x,y)$, and the distribution obtained after $t$ steps of the chain starting from $x$ is $P^t(x,\cdot)$.   We consider only irreducible chains, which have the property that for all $x,y\in \Omega$ there exists a $t$ such that $P^t(x,y) > 0$, which is equivalent to the state space being connected by the allowed transitions.  Irreducibility implies that $\pi_x > 0$ for all $x\in \Omega$, by the Perron-Frobenius theorem.  We further assume that $P$ is reversible (i.e., satisfies detailed balance), $\pi_x P_{xy} = \pi_y P_{yx}$ for all $x,y\in\Omega$.  Together with irreducibility, this implies that $R_{xy} = \pi_x^{-1/2} P_{xy} \pi_y^{1/2}$ is a symmetric matrix that is similar to $P$, so that $P$ has real eigenvalues.   Since the magnitude of these eigenvalues is at most 1, we may assume that for any chain $P$ we can analyze the closely-related lazy chain $\frac{1}{2}(I+P)$ which has all non-negative eigenvalues.  Under these conditions, $\pi$ is the unique stationary distribution of $P$ starting from any initial state $x\in \Omega$, $\lim_{t\to \infty} P^t(x,\cdot) = \pi$.
 
MCMC is an algorithm for approximately sampling from $\pi$ by simulating the transitions of $P$ for $t$ steps, sampling from $P^t(x,\cdot)$ for a randomly chosen initial states $x$. Therefore the running time of an MCMC algorithm depends on the Markov chain \textit{mixing time}, which an upper bound on the number of steps that must be simulated  to guarantee that that the output distribution is close to $\pi$ regardless of the initial state. The mixing time is defined by:
\begin{equation}\label{eq:tmixdefinition}
    t_{\text{mix}} \coloneq \min \{ t: \max_x \| P^t(x,\cdot)-\pi \|_{\text{TV}}\leq 1/4\}
\end{equation}
where $\|\mu-\nu\|_{\text{TV}} \equiv \frac{1}{2}\sum_x|\mu(x)-\nu(x)|$ is the total variation distance. The total variation distance can be used to bound the differences of expectation values of any observables and thus has an operational meaning for observable estimation procedures.

Few techniques exist that provide upper bounds on the mixing time directly, so it is often helpful to work with a related quantity called the \textit{relaxation time}. The relaxation time is defined as the inverse of the spectral gap of the of the transition matrix, that is:
\begin{equation}
    t_{\text{rel}} \coloneq \frac{1}{1-\lambda^*}
\end{equation}
where $\lambda^*$ is the second largest magnitude eigenvalue of $P$. The mixing time and relaxation times are related by the inequality:
\begin{equation}
    t_{\text{mix}} \leq t_{\text{rel}}\log\frac{4}{\pi_{\min}}\label{eq:relaxmixing}
\end{equation}
where $\pi_{\min}\coloneq \min_x\pi(x)>0$. For distributions where $\pi_{\min}$ is at most exponentially small we have that the mixing and relaxation times are related by a polynomial in the system size.

A powerful method for upper bounding the relaxation time is via the \textit{method of canonical paths}. In order to employ it we define a a graph $G= (V,E)$ (not to be confused with the interaction graph used in the definition of \eqref{eq:HeisenbergHGeneral}) on the the state space of a Markov chain transition matrix $P$ where the vertex set is $V = \Omega$, and the edge set is defined by:
\begin{equation}
    E \coloneq \{(x,y):P(x,y)>0\}.
\end{equation}
Then for each pair of states $x,y$ a path of edges from $x$ to $y$, $\gamma_{xy} = (e_1,...,e_f)$ (again not to be confused with the edges defined in \eqref{eq:SSE-loop-state-space}) such that the first entry of $e_1$ is $x$  and the second entry of $e_f$ is $y$, is chosen with the set of all paths being denoted $\Gamma$. The length of a path $f$ is defined as the total number of edges and will be denoted $|\gamma_{xy}|$. Given a choice of paths $\Gamma$ the relaxation time can be upper bounded in terms of the the \textit{congestion ratio}:
\begin{equation}
    \Phi(\Gamma) \coloneq \max_{(z,w)\in E}\frac{1}{\pi(z)P(z,w)}\sum_{\gamma_{xy}
    \ni (z,w)}\pi(x)\pi(y)|\gamma_{xy}|
\end{equation}
via the inequality:
\begin{equation}\label{eq:relaxcongestion}
    t_{\text{rel}}\leq \Phi(\Gamma).
\end{equation}

To assist in the construction of a good set of paths it is useful to use \textit{encoding functions} which are functions $\eta_{(z,w)}$ from the set of paths that pass through an edge $(z,w)$ into the state space satisfying:
\begin{equation}\label{eq:encoding-func-inequality}
    \pi(x)\pi(y)\leq C \pi(z)\pi(\eta_{(z,w)}(\gamma_{xy})).
\end{equation}
If every encoding function can be guaranteed to be at most $K$ to $1$, this leads to an upper bound on the congestion ratio of:
\begin{equation}\label{eq:congestionwithfactors}
    \Phi(\Gamma) \leq \max_{(z,w)\in E} \frac{KCl}{P(z,w)}
\end{equation}
where $l \equiv \max |\gamma_{xy}|$ is the length of the longest path.

\subsection{SPE Markov chain moves}\label{sec:actualQMCprocedure}
In this work, we consider a reversible Markov chain with Metropolis-Hastings transition probabilities~\cite{hastings1970monte} that is designed to converge to the stationary distribution \eqref{eq:piEquilibrium} or \eqref{eq:stationarywithweights}.  
Due to the simplicity of our target distribution, we can consider a simple transition rule that selects a local operator in the string $(x_1, x_2,\ldots,x_{2B})$ uniformly at random and update it to any other local operator, and the result will have non-zero weight.  In detail, each update is computed as follows:

\setlength{\leftmargini}{3.5em}
\begin{itemize}
\item[Step 1.] Pick one number from 1 to $2B$ uniformly randomly; call it $k$.  
\item[Step 2.] Choose one operator from the set of all possible operators ($\mathcal E \cup \mathcal E_{\mathrm F} \cup \mathcal V$ for the most general case) with uniform probability; call it $o$. 
\item[Step 3.] Calculate the ratio of the two configuration probabilities $\pi_{\mathrm{next}}/\pi_{\mathrm{now}}$ where $\pi_{\mathrm{now}}$ is the equilibrium probability of the current configuration and $\pi_{\mathrm{next}}$ is that of the configuration after changing the $k$-th operator in the sequence to $o$, accordingly to \cref{eq:stationarywithweights} or \cref{eq:piEquilibrium}. 
\item[Step 4.] Change the configuration to the latter with probability $\min \{1, \pi_{\mathrm{next}}/\pi_{\mathrm{now}}\}$, otherwise do nothing. 
\end{itemize}
The above 4 steps constitute a single {\it Markov Chain step} in our algorithm, and correspond to applying the stochastic matrix $P$ in the explanation of the mixing time machinery. These transitions suffice to make the chain ergodic, 
and is also not difficult to see that it is aperiodic. 
Note that all steps only require polynomial amount of floating-point numerical computation; Most importantly we do not need to explicitly calculate $\pi_{\mathrm{now}}$ and $\pi_{\mathrm{next}}$ which are both exponentially small numbers, but rather only need their {\it ratio} which can be easily be seen to be only polynomially large and efficiently computable.

\section{Proof of Main Result}\label{sec:mainproof}
In this section we establish rapid mixing of the SPE QMC for Hamiltonians of the form \eqref{eq:fullHamiltonian} on $O(1)$-bipartite graphs, using the canonical path technique we introduced in \cref{sec:mcmcbasics}.  
The SPE QMC samples the distribution defined in \cref{sec:introduceQMC} with the MCMC procedure defined in \cref{sec:actualQMCprocedure}. 
We define the set of canonical paths in  \Cref{sec:Canonical-Paths}, and establish the geometrical facts we will need about the state space of loop configurations in \Cref{sec:looptopology}.   Then the strategy is to establish a value of $C$ for the encoding inequality \eqref{eq:encoding-func-inequality}  applied to our set of canonical paths, first for the uniform weight AFM Hamiltonian \eqref{eq:SimpleHam} in  \Cref{sec:Anti-Ferro-magnet-Proof}, and then with inclusion of ferromagnetic terms, local $X$ terms, and  non-uniform weights as in \eqref{eq:fullHamiltonian} in \Cref{sec:Ferro-and-Local}.  Finally, we assemble the encoding inequality bounds along with the other factors in \eqref{eq:congestionwithfactors} and \eqref{eq:tmixdefinition} to obtain explicit bounds on the mixing time in \Cref{sec:full-upper-bound}.

\subsection{Canonical paths between configurations}\label{sec:Canonical-Paths}
We define the path between two bridge configurations $x, y \in \Omega$ as the sequence $z(t)\in\Omega$ for  $t\in[0,2B]$ where
$z(t) = (y_1,...,y_{t},x_{t+1},...,x_{2B}) \eqqcolon (y_{[t]}x_{[t+1,2B]})$
so that $z(0) = x$ and $z(2B) = y$.
The canonical path will then be the sequence of edges
\begin{align}\label{eq:canonical-path-SSE}
\gamma_{xy} = \big(z(t-1),z(t)\big)_{t=1}^{2B}.
\end{align}
The corresponding encoding function is defined to be
\begin{align}\label{eq:encoding-func-SSE}
    \eta_{(z(t-1),z(t))}(\gamma_{xy}) &= ~(x_1,...,x_{t-1},y_{t},...,y_{2B})\\
     &\eqqcolon   (x_{[t-1]}y_{[t,2B]})
\end{align}
 Recall that the goal is for the encoding function to map all canonical paths passing through the arbitrary edge $(z(t-1),z(t))$ to distinct points within the state space, $\gamma_{xy} \mapsto \eta_{(z(t-1),z(t))}(\gamma_{xy})$.   From the edge $(z(t-1),z(t))$ we have $z(t)=(y_{[t]}x_{[t+1,2B]})$, so that along with the encoded state $\eta_{(z(t-1),z(t))}(\gamma_{xy}= (x_{[t-1]}y_{[t,2B]}))$, we can uniquely recover the endpoints $x,y$ and hence the path $\gamma_{xy}$.  Hence the encoding function is invertible and therefore one-to-one.

\subsection{Bipartite Loop Topology}\label{sec:looptopology}

Here we establish some facts about the the loop representation on bipartite graphs that will prove useful. Consider the Heisenberg model defined over a bipartite graph with sublattices given by $\mathcal A$ and $\mathcal B$. Following Eq.\eqref{eq:ConsistentSeq}, for a given directed loop, all of the participating spins in the $\mathcal A$ sublattice must point in the same direction while all of the spins in the $\mathcal B$ sublattice must point in the opposite direction. This follows because the loop can only switch from one of the sublattices to the other when an operator bridge is encountered, where the spins must point in the other direction after crossing the bridge. Since by definition there are no operator bridges between sites within $\mathcal A$ or $\mathcal B$, the spins on each side of the bi-partition must always be opposite to each other if they are part of the same loop. 
This is very clear when we consider the $z$-basis configuration, but the same logic holds even if we move to the purely geometric loop representation. We can imagine a directionality in the loops, and it must always satisfy the above parity constraint due to the bipartiteness of the graph.

For the SPE considered here with open boundary conditions, this leads to strong constraints on the topology of the allowed loops, especially when the sublattices differ significantly in their number of sites. 
\begin{fact}\label{fct:LoopTermReq1}
    Let $\mathcal G = (\mathcal V, \mathcal E)$ be a bipartite graph describing an anti-ferromagnetic Heisenberg model as in \eqref{eq:HeisenbergHGeneral} with the two sides of the bi-partition denoted by $\mathcal A\subset \mathcal V$ and $\mathcal B\subset \mathcal V$. Then in the loop representation of the model, loops that begin at the top boundary in one of the sublattices of the bi-partition must either terminate at the top boundary in the other sublattice, or at the bottom boundary of the same sublattice.      
\end{fact}

This follows from how the directed loops are constructed as in \Cref{fig:antiFerroLoopExplanation}. If the loop starts in at the top boundary of $A$, it moves towards the bottom until it encounters an operator bridge, whereby it moves to the $B$ sublattice and starts moving towards the top. In this case the loop can only move down on the $A$ sublattice and up on the $B$ sublattice, which means that it must either terminate at the bottom boundary of $A$ or the top boundary of $B$. Note that this procedure is entirely deterministic once the operator string is set.

\begin{fact}\label{fct:LoopTermReq2}
    Let $\mathcal G = (\mathcal V, \mathcal E)$ be a bipartite graph describing an anti-ferromagnetic Heisenberg model as in \eqref{eq:HeisenbergHGeneral} with the two sides of the bi-partition denoted by $\mathcal A\subset \mathcal V$ and $\mathcal B\subset \mathcal V$ such that $|\mathcal B|\geq|\mathcal A|$. Then in the loop representation of the model, at least $|\mathcal B| -|\mathcal A|$ of the loops that start at the left boundary in the $\mathcal B$ sublattice must terminate at the right boundary in the $\mathcal B$ sublattice.         
\end{fact}

This follows as due to \Cref{fct:LoopTermReq1}, loops starting at the left boundary $\langle \sigma_L|$ of $\mathcal B$ can either terminate at the left boundary in $\mathcal A$ or must terminate at the right boundary $|\sigma_R\rangle$ in $\mathcal B$. Since there are at most $|\mathcal A|$ sites for the loops starting in $\mathcal B$ to terminate on the left boundary, by the pigeonhole principle, the remaining $|\mathcal B| - |\mathcal A|$ must terminate at the right boundary in the $\mathcal B$ sublattice.

\begin{lemma}\label{lem:LoopTermReq3}
    Let $\mathcal G = (\mathcal V, \mathcal E)$ be a bipartite graph describing an anti-ferromagnetic Heisenberg model as in \eqref{eq:HeisenbergHGeneral} with the two sides of the bi-partition denoted by $\mathcal A\subset \mathcal V$ and $\mathcal B\subset \mathcal V$ such that $|\mathcal B|\geq|\mathcal A|$. Then in the loop representation of the model, consider a "cut" as in \Cref{fig:encoding-inequality-example} in the spatial direction. For any given configuration and any cut the amount of loops passing through the cut can at most differ by $2|\mathcal A|-1$. \end{lemma}
\begin{proof}
    From  \Cref{fct:LoopTermReq2}, every cut must have at least $|\mathcal B| -|\mathcal A|$ loops passing through it as these loops must get from the left boundary to the right boundary. These loops pass through a minimum of $|\mathcal B| -|\mathcal A|$ sites, leaving at most $2|\mathcal A|$ sites for the remaining loops to pass through. On one hand every site could correspond to its own loop, yielding $|\mathcal B| +|\mathcal A|$ loops passing through. On the other extreme, all of the $2|\mathcal A|$ sites could be connected in a single loop, which would yield $|\mathcal B| -|\mathcal A| +1$ loops. The difference of the number of loops in these two extremal configurations is $2|\mathcal A| -1$.  
\end{proof}

\subsection{The Unweighted Anti-Ferromagnet}\label{sec:Anti-Ferro-magnet-Proof}

We now turn towards using these topological restrictions to show that the constant $C$ in \eqref{eq:encoding-func-inequality} can be upper bounded, implying rapid mixing when this technique is applied to the unweighted anti-ferromagnetic model \eqref{eq:SimpleHam} defined over $O(1)$-bipartite graphs. 

\begin{lemma}\label{lem:loopdiffbound}
    For the SPE MCMC on an unweighted bipartite anti-ferromagnetic Heisenberg Hamiltonian with the smaller part of the bipartition labeled as the set $\mathcal A$, the constant $C$ in equation \eqref{eq:encoding-func-inequality} can be upper bounded by $2^{4|\mathcal A| - 2}$. 
\end{lemma}
\begin{proof}
If we insert \eqref{eq:encoding-func-SSE} into \eqref{eq:encoding-func-inequality} for the path $\gamma_{xy}$ and edge $(z(t-1),z(t))$ as in \eqref{eq:canonical-path-SSE} obtain the inequality:
\begin{align}\label{eq:inequality-encoding-function-SSE}
    &\pi(x_{[t-1]}x_{[t,2B]})\pi(y_{[t-1]}y_{[t,2B]})\\
    \leq &C \pi(y_{[t-1]}x_{[t,2B]})\pi(x_{[t-1]}y_{[t,2B]}).
\end{align}

As demonstrated in \Cref{fig:encoding-inequality-example}, we now divide the loops in the configurations appearing above into into \textit{internal} loops, which pass through the cut between operators $t-1$ and $t$ in the sequence, and \textit{external} loops which do not pass through the cut. The external loops will appear in both sides of the inequality and will cancel out, meaning only the internal loops can contribute to the inequality. Since by \Cref{lem:LoopTermReq3} the max difference in the number of internal loops for any cut is upper bounded by a $2|\mathcal A|-1$, the worst case scenario is when $x$ and $y$ both have $2\mathcal|A|-1$ more loops than $y_{[t-1]}x_{[t,2B]}$ and $y_{[t-1]}x_{[t,2B]}$, which means that \eqref{eq:encoding-func-inequality} can always be satisfied if:
\begin{align}
    C = 2^{4|\mathcal A| - 2}.
\end{align}
\end{proof}

\begin{figure}[t]
    \centering
    \includegraphics[scale=0.085]{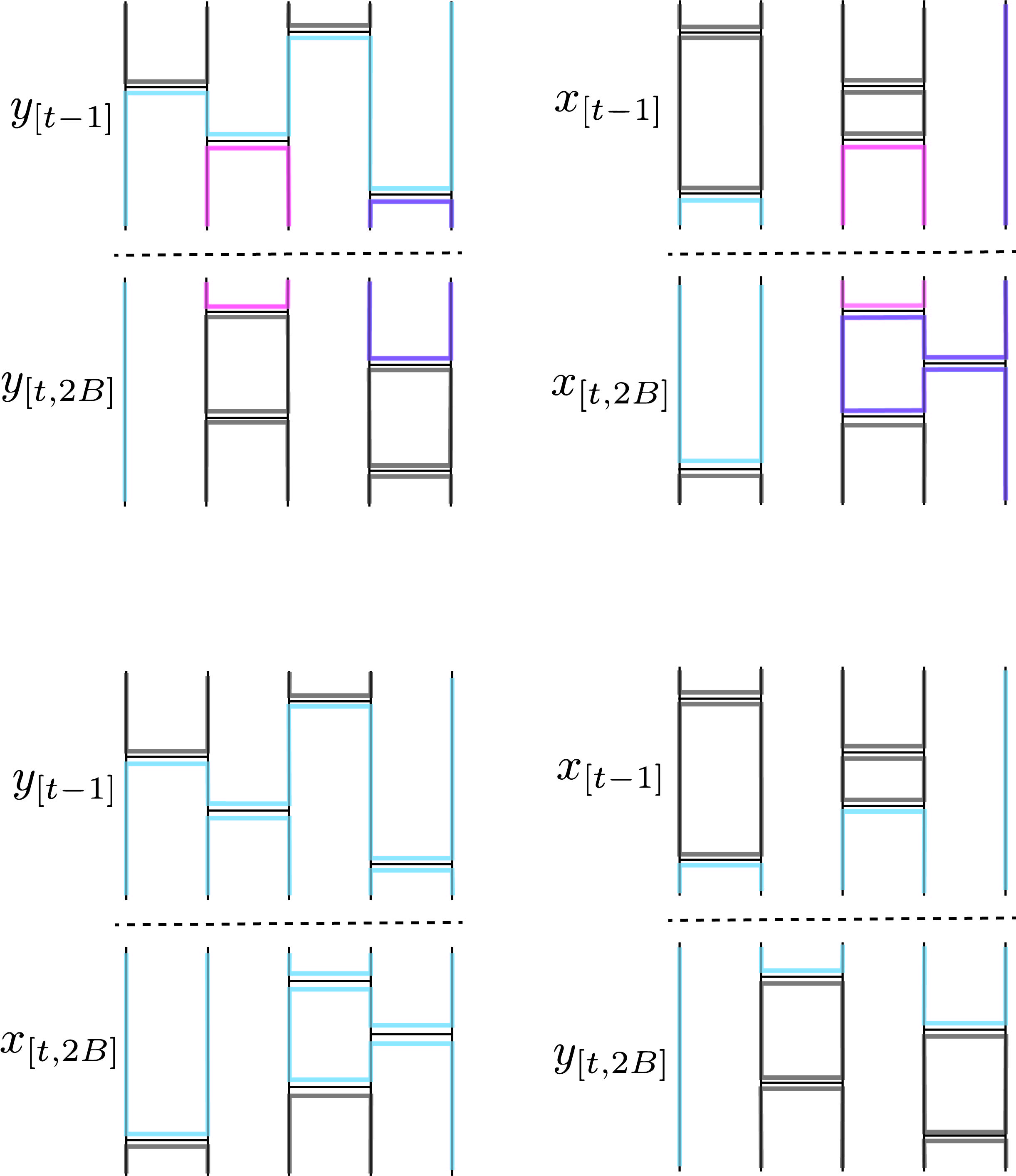}
    \caption{An example of the the configurations appearing in the inequality \eqref{eq:inequality-encoding-function-SSE}. Here the external loops are colored grey while any loops that pass through the cut between operator $t-1$ and $t$ are given unique colors in each configuration. Notice that each of the grey loops appears on both the top pair and the bottom pair of configurations, meaning that their contributions will cancel out in \eqref{eq:inequality-encoding-function-SSE}.}
    \label{fig:encoding-inequality-example}
\end{figure}

\subsection{Extension to cases with Ferromagnetic interactions, Local Fields, and nonuniform coefficients}\label{sec:Ferro-and-Local}

\paragraph{Ferromagnetic Interactions.}The extension to the case with included ferromagnetic terms as in \eqref{eq:FMterms1} is straightforward. Recall that the ferromagnetic terms come from the modified edge set $\mathcal E_F$ which only contains edges from $\mathcal A$ to $\mathcal A$ and $\mathcal B$ and $\mathcal B$ sites respectively. This means that a modified form of  \Cref{fct:LoopTermReq1} holds. More precisely we have: 
\begin{fact}\label{fct:FerroLoopTermReq1}
    Let $\mathcal G = (\mathcal V, \mathcal E)$ be a bipartite graph with the two sides of the bi-partition denoted by $\mathcal A\subset \mathcal V$ and $\mathcal B\subset \mathcal V$, and let $E_F\subset \{(u,v):u,v\in \mathcal A \text{ or } u,v\in \mathcal B\}$, where $\mathcal E$ represents the anti-ferromagnetic interactions and $\mathcal E_F$ represents the ferromagnetic interactions. Then in the loop representation of the model, loops that begin at the top boundary in one of the sublattices of the bi-partition must either terminate at the top boundary in the other sublattice, or at the bottom boundary of the same sublattice.  
\end{fact}

This follows by similar reasoning to \Cref{fct:LoopTermReq1}, noting that the ferromagnetic terms as described in \Cref{fig:FerroTerms} do not reverse the direction of the directed loops, and since by definition they only appear between sites that are in the same sublattice, whenever one is encountered both the loop propagation direction and side of the bi-partition are unchanged.  

Since \Cref{fct:LoopTermReq2} and thus \Cref{lem:LoopTermReq3} and the final result are logical consequences of \Cref{fct:LoopTermReq1}, it follows that the addition of ferromagnetic terms on $O(1)$ and $O(\log(n))$ bipartite graphs does not affect the rapid mixing results.

\paragraph{Local Fields.}  When we consider the case with local $X$ fields as discussed in  \Cref{subsubsec:TF}, the notation for a configuration $x = (x_1,...,x_{2B})$ now implicitly includes various $x_j$ that correspond to applying an $X$ operator at some site.  We consider the same set of canonical paths $\{\gamma_{xy}\}$ as in the previous section, updating the operators (including local $X$ operators) in sequence from left to right.  

 For any configuration $z \in \Omega$, let $q(z)$ be the same sequence of operators as $z$ except that all the 1-local $X$ operators are removed.   Since our QMC method uses a fixed value of $B$, $q(z)$ is not in $\Omega$ in general because it contains fewer than $2 B$ operators, but $q(z)$ is still totally a valid configuration except for that point. 
 Our strategy for the analysis including local fields is based on applying  \Cref{lem:LoopTermReq3} to $q(z)$, and then considering the effect of restoring the 1-local $X$ operators to $q(z)$ one by one.   We emphasize that this does not alter the left-to-right canonical paths (i.e. we don't treat the 1-local $X$ operators differently along the path), and that this is only a tool for analyzing the weights.   In this process we will consider configurations with fewer than $2 B$ operators, although the QMC itself never visits such states.

\begin{fact} Let $z = (z_1,...,z_{b})$ be an arbitrary configuration and consider inserting a 1-local $X$ operator at position $j$ in the sequence (labeled as the operator $z_j$), to form a new sequence $z'$ of length $b+1$.    The relation between the total number of loops $L(z')$ in the new configuration and the number of loops $L(z)$ in the previous configuration falls into one of two cases:
\setlength{\leftmargini}{5em}
\begin{itemize}
    \item[Case 1:] If $z_j'$ acts on a site that is part of \\
    a closed loop with no $X$ operators \\
    in $z$, then $L(z') = L(z)$.  

    \item[Case 2:] Otherwise, $L(z') = L(z) + 1$.
\end{itemize}
\end{fact}
In the first case, adding an $X$ operator to a closed loop that contains no other $X$ operators, results in turning that closed loop into an open loop whose endpoints coincide at the location of this new operator - this leaves the total number of loops unchanged.  In all other cases, inserting an $X$ operator divides a single loop (really an open string) into two loops, thereby increasing the number of loops by one.    
    
  Now we define $k(z)$ to be the number of 1-local $X$ operators in $z$, and $\ell(z)$ be the number of closed loops in $q(z)$ that contain at least one $X$ operator in $z$.   Now consider going from $q(z)$ to $z$ by restoring the 1-local $X$ operators one at a time, and applying Fact 4 at each step, we have the relation

\begin{equation}L(z) = L(q(z)) + k(z) - \ell(z) \label{eq:1localloops}\end{equation}
We now seek to upper bound the congestion for this system of canonical paths, by establishing an encoding inequality as before,
\begin{equation}\pi(x)\pi(y) \leq K \pi(z)\pi(\eta) \label{eq:1localcongestion}\end{equation}
where $z = y_{[t-1]}x_{[t,2B]}$ and $\eta = \eta_{(x,y)}(z)=x_{[t-1]}y_{[t,2B]}$.  By \Cref{lem:LoopTermReq3}, this inequality holds when all $X$ operators are deleted,
\begin{equation} \pi(q(x)\pi(q(y)) \leq C \pi(q(z))\pi(q(\eta)) \label{eq:1localbasecongestion}\end{equation}
for $C = 2^{4|A|-2}$.  Now \eqref{eq:1localcongestion} is equivalent to
$$2^{L(x) + L(y) - L(z)-L(\eta)} \leq K$$
Applying \eqref{eq:1localloops} and \eqref{eq:1localbasecongestion} this becomes
$$2^{[(k(x) + k(y)) - (k(z)+k(\eta))]-  \ell(x) - \ell(y) + \ell(z) + \ell(\eta)} \leq C K$$
The part of the exponent in square brackets is at most 1 by the design of the canonical paths (the number of $X$ operators in $(x,y)$ and $(z,\eta)$ differs by at most 1).  Lastly we need to compare $\ell(z) + \ell(\eta)$ (the number of closed loops in $q(z)$ and $q(\eta)$ on which $X$ terms act  in $z$ and $\eta$) to $\ell(x) + \ell(y)$.   Noting the signs in the exponent terms, the ``bad case'' occurs when $\ell(z) + \ell(\eta)$ is larger than $\ell(x) + \ell(y)$.  This bad case corresponds to many $X$ operators in $x$ and $y$ that fall under case 2 of fact 4, now instead falling under case 1 of fact 4 in $z$ and $\eta$.   In other words, the 1-local $X$ terms can only reduce $L(z) + L(\eta)$ in comparison with $L(x) + L(y)$ by at most the number of new closed loops in $z,\eta$ that do not appear in $x,y$.   

By \Cref{lem:LoopTermReq3}, the number of such loops can be bounded in terms of the size of the $O(1)$ sublattice.   Since $\ell(x) + \ell(y)$ and $\ell(z) + \ell(\eta) $ count subsets of loops from the corresponding configurations $q(x),q(y)$ and $q(z),q(\eta)$, these loops can be divided just as before into external loops (which do not intersect with the cut at time $t$ and internal loops which do intersect with the cut.   Consider an example of an external loop that occurs before the cut at time $t$ in the configuration $q(x)$.  In $x$ this external loop will in general be cut by some local $X$ operators, but these are also external in the sense that by time $t$, the loop from $q(x)$ appears in $q(z)$, and the $X$ operators acting on that loop in $x$ also appear in $z$, increasing the number of loops in each configuration by the same amount so that these contributions cancel. 

Therefore the difference in the number of loops created by $X$ operators can only come from closed internal loops across the cut in $q(z),q(\eta)$ that do not appear in $q(x),q(y)$, and so by \Cref{lem:LoopTermReq3} we have
$$(\ell(z) + \ell(\eta)) - (\ell(x) + \ell(y)) \leq 2|A| $$
and therefore we can satisfy the congestion bound \eqref{eq:1localcongestion} with $K = C^2$.  

\paragraph{Non-uniform Coefficeints.}  The generalization of each case to include non-uniform couplings on the local Hamiltonian terms is straightforward due to the simple product form of the weights in \eqref{eq:loopzwithweights},
$$
\pi(z) = \frac{2^{L(x)}}{Z} \prod_{k=1}^{2B}T(x_k).
$$
Recall that $T(x_k)$ is the coupling strength $w_{ij}$ ($J_{ij}$) if $x_k$ is a 2-local AFM (FM) term acting on qubits $i,j$, or $T(x_k) = g_i$ if $x_k$ is a 1-local $X$ operator acting on qubit $i$.   The main observation is that this product of couplings cancels on both sides of the encoding inequality for our canonical paths,
$$\pi(x) \pi(y) \leq C \pi(z) \pi(\eta)$$
Every term $T(x_k), T_(y_k)$ on the left side cancels exactly with a term $T(z_k),T(\eta_k)$ on the right side.   The inclusion of non-uniform Hamiltonian couplings has no effect on the encoding inequality, but it does have a small effect on the transition probabilities as we describe in the next section.

\subsection{Upper bound on the congestion and mixing time.}\label{sec:full-upper-bound}
We now return to the full bound on the congestion \eqref{eq:congestionwithfactors}, 
\begin{equation}
    \Phi(\Gamma) \leq \max_{(z,w)\in E} \frac{KCl}{P(z,w)} \leq \frac{K C l}{P_{\min}}
\end{equation}
where $l \equiv \max |\gamma_{xy}|$ is the length of the longest path, and $P_{\min}$ is the smallest non-zero transition probability along any edge.  Our canonical paths have length $l \leq 2 B$.  The Hamiltonian contains at most $3 N$ local terms, so the probability of proposing to update a particular operator at a particular location $1,...,2 B$ is $1/((2 B)(3 N))$.   The update of this operator will reduce the number of loops by at most 1, thereby reducing the weight by 1/2, and if the couplings are non-uniform the Metropolis rule will accept with probability proportional to $T_{\min} / T_{\max}$, and so the smallest transition probability is
$$P_{\min} = \frac{1}{6 B N} \left(\frac{T_{\min}}{T_{\max}} \right)$$
which indicates that the ratio of any two non-zero couplings in the Hamiltonian should be at least $1/\textrm{poly}(N)$.   Inserting these factors into \eqref{eq:congestionwithfactors}, 
$$\Phi(\Gamma) \leq 12 B^2 N \left(\frac{T_{\max}}{T_{\min}} \right) 2^{8 |\mathcal{A}| - 4}$$
This bounds the relaxation time $t_{\textrm{rel}}$ by \eqref{eq:relaxcongestion}.   To bound the mixing time by \eqref{eq:relaxmixing} we also need the bound $\pi_{\min} \geq (2 T_{\min})^{-2  B} $ (which follows from the maximum number of loops in any configuration being bounded by the number of operators $2B$), so we have

\begin{equation}
t_{\textrm{mix}} \leq 24 B^3 N \left(\frac{T_{\max}}{T_{\min}} \right) 2^{8 |\mathcal{A}| - 4} \log(8 T_{\min}),
\end{equation}
which gives us rapid mixing for $|\mathcal A| = O (1)$ under the mild condition that $T_{\max}/T_{\min}$ is at most polynomially large. 
Since the smallest coefficient of the Hamiltonian is essentially what determines the energy scale of the Hamiltonian, this condition can be reread as the largest coefficient being only polynomially large after rescaling. 
This is a reasonable condition both for physical and complexity-theoretic considerations.

\subsection{Topological Viewpoint}

Stepping back to the big picture, our proof relies on topological features of the loop representation that hold on $O(1)$-bipartite graphs.  This perspective is summarized in \Cref{fig:topology}, where the topological structure of the loops is emphasized in (c). 
The crucial point in our proof for \Cref{lem:loopdiffbound} is to notice that only the loops touching the central ``cut'' in the canonical path plays a nontrivial role in the analysis. 
This is shown in (d) with the center cut as the dotted line, and coloring the irrelevant loops (not touching the cut) in gray. 
This gives another way to pictorially see that the number of proper loops touching the cut cannot exceed $|\mathcal A|$ thus giving the bound in \Cref{lem:loopdiffbound}. 

As shown in \Cref{fig:topology} (e), even with $X$-fields, the essence of the proof remains the same. The crucial point is still the fact that you can have at most $|\mathcal A|$ proper loops, and this does not change even when we have $X$-operators, because all they can do is to cut the loops. 
The bounding argument holds independently of whether the central loops are terminated at the true boundaries $\langle \sigma_L|$ and $|\sigma_R\rangle$ or at an $X$-operator. 

It is interesting to see that the analysis of this simple canonical path naturally leads to topological considerations as the above. The simplicity of star-like $O(1)$-bipartite graphs comes from the fact that there {\it cannot} be any highly nontrivial topologies in these systems. 
The ground state of the Heisenberg model on a star graph is exactly mappable to a thermal state of a 1-dimensional classical Potts model as we show in \cref{sec:pottsduality}, and thus can in principle be analyzed in a totally classical way. We can say that the lack of nontrivial topologies is the underlying factor that enables such mappings to a rather simple classical model. 

\begin{figure}[t]
\centering
\includegraphics[width=8cm]{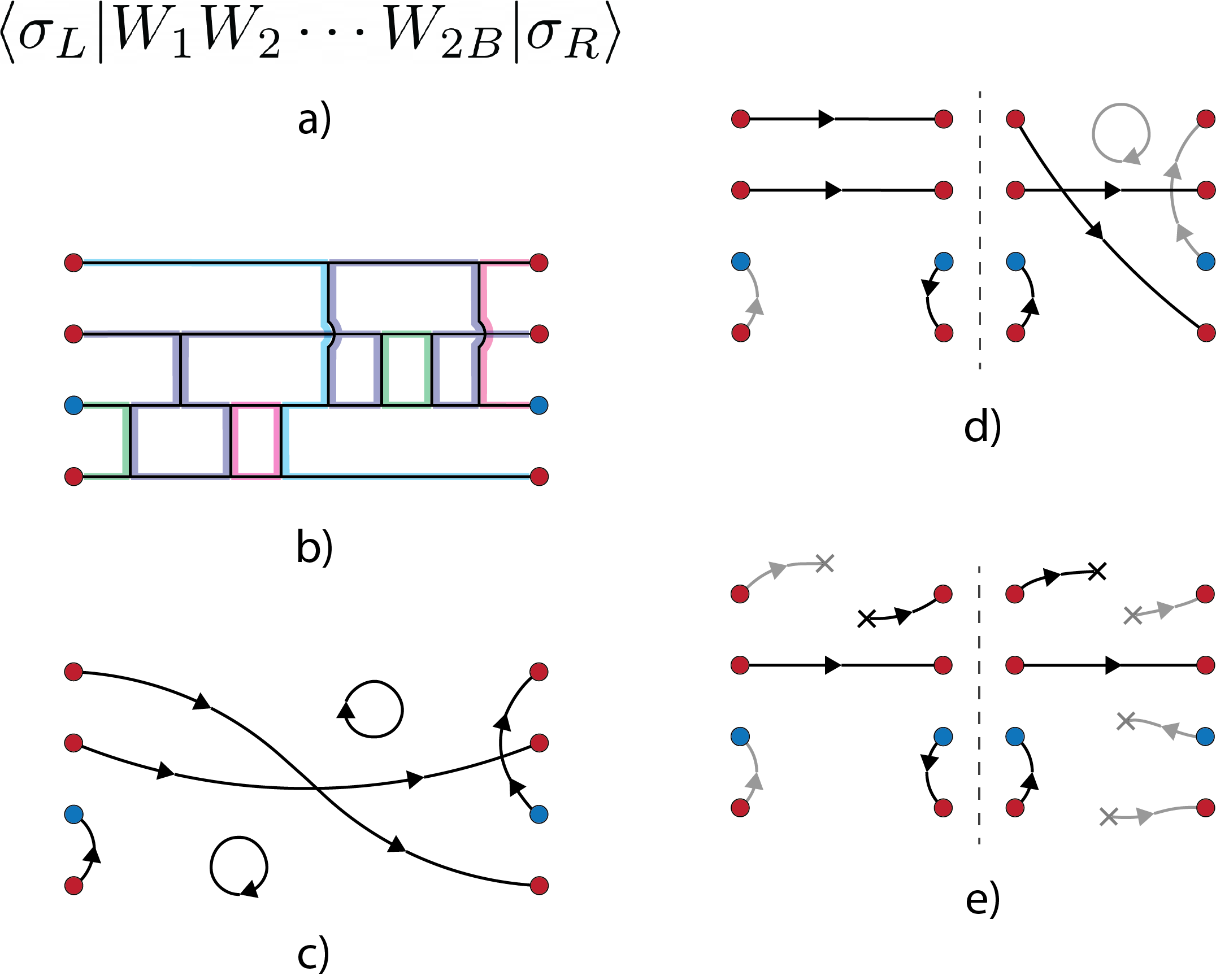}
\caption{Five ways to view the sampled configurations. 
(a) The very basic operator sequence as written in \cref{eq:ConsistentSeq}. 
(b) The diagramatic depiction of the loop representation as in \cref{eq:SSE-loop-state-space}. The bipartiteness of the graph is represented with blue and red vertices. 
(c) The Topological way of viewing the loop representations. The arrows are guide to the eyes with the rule that red goes from left to right, and opposite for blue. 
(d) The way such a topologically depicted loop configuration will look like under the canonical path we consider, cut in the middle. 
(e) The same thing as (d) but for the case with 1-local $X$-terms. For (d) and (e) we show the loops and line segments not touching the cut in gray, indicating their irrelevance to \cref{eq:inequality-encoding-function-SSE} since they easily cancel out.}\label{fig:topology}
\end{figure}

\section{Discussion and Open Problems}\label{sec:discussion}
We have established that the ground state energy of $O(1)$-bipartite antiferromagnetic Heisenberg ground states can be computed using a quantum Monte Carlo method in $\textrm{poly}(N)$ time. 
Our result extends to systems with nonuniform coefficients, staggered local $X$-fields, and ferromagnetic interactions within each bipartite sublattice, as shown in the previous section. 
Interestingly, the latter two only serve to enhance the antiferromagnetic order. 
In other words, both of the extensions only lead to larger values of 
$\langle O_{\textrm{N\'eel}}^2\rangle$ where 
\begin{equation}\label{eq:Neel} O_{\textrm{N\'eel}} = \frac{1}{N} \left(\sum_{i \in \mathcal{A}} X_i - \sum_{i \in \mathcal{B}} X_i \right)\end{equation}
is called the N\'eel order parameter. 

In the following, we discuss the outlook for extending our results to a broader class of models and examine the limitations of our current proof technique. 
These challenges seem to highlight an intriguing connection between phases of matter, computational complexity, and hardness of proofs.  
In \Cref{sec:phaseappendix} and \Cref{sec:vbs} we further discuss  the conjectured relationships between QMC mixing times and the quantum phases considered for Heisenberg models on more general bipartite graphs, the N\'eel phase and the valence-bond solid (VBS) phase.   

\paragraph{Extensions and Conservation of Difficulty.}
As we have mentioned in the introduction, the general Hamiltonian \cref{eq:GeneralHam} we have shown rapid mixing of QMC for also exactly coincides with the class of Hamiltonians where exact diagonalization only takes polynomial time thanks to the Lieb-Mattis theorem (as explained in \cref{app:LiebMattisTheorem} ), despite not using this property in our proof. In any case, our motivation is to initiate the analysis of QMC for general bipartite Heisenberg AFM, and so here we describe three natural extensions that would evade known cases of polynomial-time exact diagonalization by the Lieb-Mattis theorem and other methods, motivating them as potential next steps for QMC mixing time analysis.

The first case is when we have $ O(N)$ vertices on both sides of the bipartite graph. In this case, not only does our proof fail, but we can also explicitly construct examples where our system of canonical paths goes through configurations $z$ and $\eta$ that have combined weight exponentially smaller compared to $x$ and $y$. 
For example, in \Cref{fig:badconfigs} (a-1), we show a case where the interaction graphs is a cycle, and two particular SPE configurations $x$ and $y$. 
If we construct the canonical paths in the same way as in \cref{sec:Canonical-Paths}, then in the intermediate step shown in \Cref{fig:badconfigs} (a-2), we have configurations $z$ and $\eta$ that both have $\sim N/2 =O(N)$ smaller number of loops compared to $x$ and $y$. 

\begin{figure}[h!]
\centering
\includegraphics[width=7.8cm]{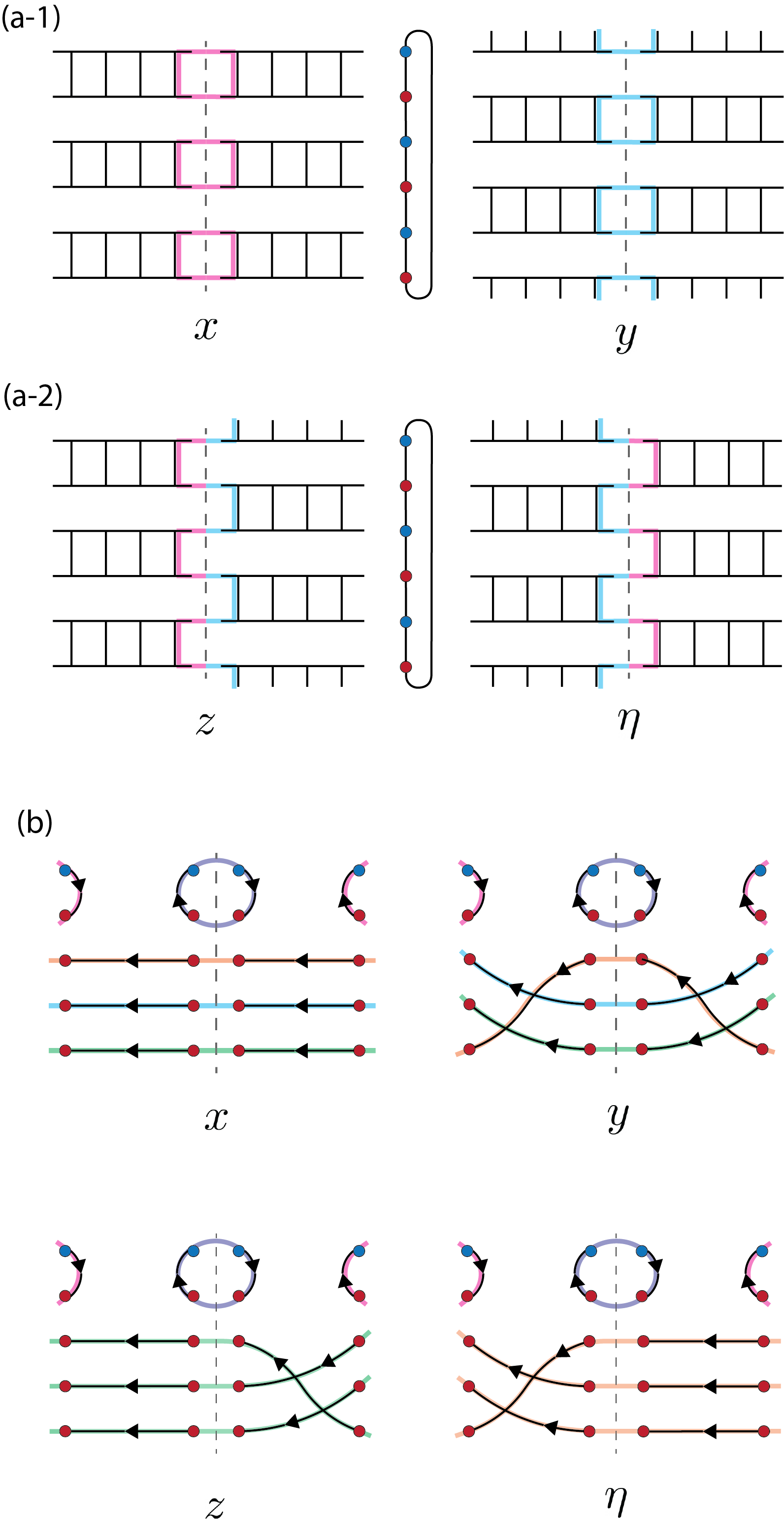}
\caption{Configurations that lead to the failure of our ``left-to-right'' canonical paths in some attempted extensions. (a-1) Two configurations $x$ and $y$ that do not work for the cycle graph where $|\mathcal A| = |\mathcal B|=O(N)$. There are $O(N)$ loops that touch the center cut. (a-2) In the middle of this path, there will be configurations $z$ and $\eta \equiv \eta_{e}(\gamma_{xy})$ that only have a single loop touching the center cut, giving $O(N)$ smaller loop numbers, thus leading to exponentially smaller weights $\pi(z)\pi(\eta) \ll \pi(x)\pi(y)$.
(b) Failure example for the star graph when we have nonzero temperature i.e., a periodic boundary condition. The configurations are represented in the topological manner as in \cref{fig:topology} for visual clarity, and actual configurations with operator/bridges that realize this topology is also not hard to construct. 
}\label{fig:badconfigs}
\end{figure}

It should be noted that the failure of the particular canonical paths we considered here (and in the subsequent examples) does not imply that the SPE Markov chain is slow mixing for these examples, only that a substantially different analysis would be needed.   The existence of a low congestion system of canonical paths is only a sufficient condition for rapid mixing, but if the chain is rapidly mixing then there necessarily exists a low congestion multicommodity flow (which generalizes the canonicaal paths method to take multiple paths between each pair of configurations) \cite{sinclair1992improved}, which motivates our examination of various ``bad paths'' that could occur in these extensions.

The next extension would be to simulate thermal states of $O(1)$-bipartite graphs above zero temperature using the standard SSE QMC, where the significant difference for the mixing time analysis is the presence of periodic boundary conditions in the imaginary time direction. 
This becomes essential for thermal simulations because of the trace in the definition of $\mathcal Z$. 
In this case with periodic boundary conditions, loops that touch the boundary then wrap around it (potentially multiple times) until they eventually close.  Since \Cref{lem:LoopTermReq3} no longer applies without an open boundary, substantial modification to the canonical paths will become necessary. 
Once again we can explicitly construct cases where the straightforward canonical path fails, as shown in \Cref{fig:badconfigs} (b). 
This case is perhaps surprising since increasing the temperature often  makes the sampling problem easier \cite{crosson2020classical, crosson2021rapid,lubetzky2012critical,levin2017markov}. 
Therefore, we could reasonably expect that the Markov chain is rapidly mixing even for closed boundary conditions, and proving this even for the star graph ($|\mathcal A|=1$ case) at non-zero temperature would be a significant extension.

The third potential direction for generalizing our results would be to simulate the ground state with the simultaneous presence of 1-local fields in both the $X$ and $Z$ directions. 
If we further assume non-uniform weights for these terms, no basis rotation can allow the 1-local fields to align in a single direction.  
The handling of 1-local $Z$ terms within the SSE / SPE method is not described in \Cref{sec:introduceQMC}; one possible way is to have local $Z$ operators together with the identity so that when they are present on a loop, the weight of the loop reduces from $2$ to $1$. This rapidly leads to counterexamples for the simple canonical paths we employ (in fact the encoding inequality becomes exponentially bad when our canonical paths are applied to a uniform field in the $Z$ direction).

Remarking again on the coincidence between the regime of our analysis and the applicability of exact diagonalization, this may point to an unknown mathematical connection between QMC dynamics and the dimensionality or entanglement in the system, even when these properties are not exploited directly.   It may also evoke the ``{\it law of conservation of difficulty}'' \cite{tasakiblog,taoblog2,emails,mulmuley2011p}, an empirical folklore that states the difficulty of obtaining a proof tends to be conserved regardless of the method.   However,  while the exact diagonalization based approach will inevitably have inefficient scaling for general bipartite {\sf QuantumMaxCut} even for the simple square lattice, we have strong empirical evidence of rapid mixing of quantum Monte Carlo for such cases with bipartite lattices \cite{harada2013possibility,takahashi2024so5}. 
If one aims to answer the 4th open problem of \cite{gharibian7faces}: 
\begin{quote}
\hspace{-8mm}
\begin{minipage}{1.2\linewidth}
\centering  {\it What is the complexity of {\sf QuantumMaxCut} on bipartite graphs?}
\end{minipage}
\end{quote}
in the positive direction (containment in {\sf BPP}), 
then quantum Monte Carlo sampling is the only current method that appears to be demonstratively efficient in practice. 
One future direction for proving stronger rapid mixing results would be to use more detailed information about the ground state.  
In \Cref{sec:phaseappendix}, we discuss this possibility in terms of assuming that the ground state is in the N\'eel phase, where $\langle O_{\textrm{N\'eel}}^2\rangle$ converges to a nonzero constant as we increase the system size in a systematic way.

\paragraph{Acknowledgements}
J.T. thanks Chaithanya Rayudu, Naoto Shiraishi, Manaka Okuyama, Yoshihiko Nishikawa, Kai Nakaishi, Tota Nakamura, Naoki Kawashima, Hosho Katsura, Hal Tasaki, and Anders Sandvik for helpful discussions. S.S. also thanks Chaithanya Rayudu for insightful discussions.

J.T. acknowledges support from the U.S. National Science Foundation under Grant No. 2116246, the U.S. Department of Energy, Office of Science, National Quantum Information Science Research Centers, and Quantum Systems Accelerator. 

This material is based upon work supported by the U.S. Department of Energy, Office of Science, National Quantum Information Science Research Centers, Quantum Science Center (QSC). S.S. was funded by the QSC to perform the analytical calculations and to write the manuscript along with the other authors
and was supported by the Laboratory Directed Research and Development program of Los Alamos National Laboratory (LANL) under project number 20230049DR and also acknowledges support from US Department of Energy, Office of Science, Office of Advanced Scientific Computing
Research, Accelerated Research in Quantum Computing program.

\bibliographystyle{plainnat}
\bibliography{mybibliography}

\onecolumn\newpage
\appendix

\section{Lieb-Mattis theorem and symmetry-reduced exact diagonalization}\label{app:LiebMattisTheorem}
The AFM Hamiltonian \eqref{eq:HAFM} commutes with $Z_{\mathrm{tot}} \coloneq \sum_{i =1}^N Z_i$, and so $H_\textrm{AFM}$ has a basis of eigenstates that are each also eigenstates of $Z_{\mathrm{tot}}$.  Furthermore, the off-diagonal matrix elements $\langle \sigma |H_\textrm{AFM}|\sigma^{\prime}\rangle$ in the $Z$ basis are non-zero only if $\sigma,\sigma^{\prime}$ are eigenstates of $Z_{\mathrm{tot}}$ with the same eigenvalue. 
We can define the total spin $S_{\mathrm{tot}}$ via $\hat{S}^2\coloneq X_{\mathrm{tot}}+Y_{\mathrm{tot}}+Z_{\mathrm{tot}}\eqqcolon 4S_{\mathrm{tot}}(S_{\mathrm{tot}}+1)$ since $\hat{S}^2$  also commutes with the Hamiltonian, and has a definite value. 
Lieb and Mattis combined this observation with the Perron Frobenius theorem and variational reasoning to deduce the following.

\paragraph{Theorem (Lieb-Mattis)~\cite{lieb1962ordering,tasaki2020physics}}  For any spin-1/2 Hamiltonian $H_\textrm{AFM}$ consisting of Heisenberg antiferromagnetic terms on a bipartite lattice with bipartite subgraphs $\mathcal{A}, \mathcal{B}$, the ground states of $H_\textrm{AFM}$ have a total spin of $S_{\textrm{tot}}=\Bigl||\mathcal{A}| - |\mathcal{B}|\Bigr|/2$, with the degeneracy  $\Bigl||\mathcal{A}| - |\mathcal{B}|\Bigr|+1$. 
 The degenerate ground states are each spanned by $Z$-basis states with Hamming weights $K=\min\{|\mathcal A|, |\mathcal B|\}, \min\{|\mathcal A|, |\mathcal B|\}+1, \ldots, \max\{|\mathcal A|, |\mathcal B|\}$ respectively. 
\vspace{10pt}

In the case of a size $N$ star-like bipartite graph with a sublattice of size $M < N/2$, this implies that 
it is sufficient to consider only $Z$-basis states with 
Hamming weight equal to $M$, and so the dimension of the subspace is $\sum_{k \leq M} {N \choose k} = O(N^M)$.  This implies that the ground state energy of $H_\textrm{AFM}$ can be found in polynomial time by exact diagonalization if $M=O(1)$.

This solution by reduced dimensionality and exact diagonalization also applies with the inclusion of ferromagnetic terms on either bipartite sublattice (these terms also commute with $S_z$, and the Lieb-Mattis theorem has a direct extension for such terms~\cite{tasaki2020physics}).   When we include 1-local $X$ fields, the Hamiltonian matrix in the $Z$ basis is an irreducible matrix, and so it has a unique ground state with support on all $Z$ basis strings (therefore the reduced dimensionality no longer holds in the $Z$ basis).   However, due to the $\mathrm{SU}(2)$ symmetry of $H_\textrm{AFM}$, we can apply a Hadamard rotation to go from $H_\textrm{AFM} + H_X$ to $H_\textrm{AFM} + H_Z$ and once again find the ground energy by exact diagonalization with a reduced dimension.  

Therefore the ground state energy of all the models we consider can be found in polynomial time using the Lieb-Mattis theorem and exact diagonalization.   The one task that our QMC algorithm can do that cannot obviously be replicated by these methods is to approximately sample from the ground state of $H_\textrm{AFM} + H_X$ in $Z$ the basis.

\section{Scaling of $B$ for Ground State Energy Estimation}\label{sec:scalingOfB}
In this section we establish the scaling of $B$ for the open-boundary condition SSE method based on \eqref{eq:expectation}.   We will assume the Hamiltonian has all negative eigenvalues $0 > E_{2^N-1} \geq \ldots \geq E_1 \geq E_0$ as in \eqref{eq:HeisenbergHGeneral}, so that $A = -H$ is a positive semidefinite operator with eigenvalues $\lambda_0 > \lambda_1 \geq ... \geq \lambda_{2^N-1} > 0$, with $\lambda_i = -E_i$.  Let $\epsilon > 0$ and define a projector onto states within $[E_0,E_0 + \epsilon]$, and a projector onto all other energy eigenstates:

$$ \Pi_{[E_0, E_0 + \epsilon]} = \sum_{E_i \in [E_0,E_0+\epsilon]} |E_i \rangle \langle E_i |  \quad , \quad \Pi_{\textrm{other}} = \sum_{E_i > E_0+\epsilon} |E_i \rangle \langle E_i |$$
Define the normalized state
$$
|M_B\rangle = \frac{(-H)^B|+^N\rangle}{\sqrt{ \langle +^N |(-H)^{2B}|+^N \rangle}}= \frac{A^B|+^N\rangle}{\sqrt{\langle +^N |A^{2B}|+^N \rangle}}
$$
We will first derive the scaling of $B$ necessary to achieve $\langle M_B | \Pi_\textrm{other} | M_B\rangle < \delta$, for any $\delta > 0$.  Analyzing the numerator,
$$\langle +^N | A^{B}\Pi_\textrm{other}A^B| +^N\rangle \leq (2^N -1)(\lambda_0 - \epsilon)^{2B} $$
Now examining the denominator, 
$$\langle +^N | A^{2B} | +^N\rangle \geq |\langle \psi_0 | +^N\rangle|^2 \lambda_0^{2B}$$
Defining $w = |\langle \psi_0 | +^N\rangle|^2$, we have
$$\langle M_B | \Pi_\textrm{other} | M_B\rangle \leq (2^N - 1) \left(1 - \epsilon/\lambda_0 \right)^{2B}w^{-1} \leq 2^N w^{-1} e^{-2 B \epsilon/\lambda_0}\leq e^{N - 2 B \epsilon /\lambda_0 - \ln w}$$
setting the final exponent $N - 2 B \epsilon /\lambda_0$ equal to $\ln \delta$ and solving for $B$ yields,
$$B = \frac{|E_0|}{2 \epsilon} (N + \ln (\delta w)^{-1})$$
Recall that all the eigenvalues of $H$ are negative, so $|E_0| = \|A\| = \|H\|$, and so the bound can be equivalently expressed in terms of the operator norm of $H$.  

Since $\Pi_{[E_0,E_0 + \epsilon]} + \Pi_\textrm{other} = \mathbbm{1}$, taking $B$ as above will suffice to achieve $\langle M_B | \Pi_[E_0,E_0 + \epsilon] | M_B \rangle \geq 1 - \delta$.  Therefore to estimate the energy to within precision $\epsilon$ using the local Hamiltonian terms, we can take $\delta = \epsilon/\|H\|$.   

As for the overlap $w = |\langle \psi_0 | +^N\rangle|^2$, the Hamiltonian is stoquastic and so $\psi_0$ is a normalized state with all non-negative amplitudes in the computational basis.   Therefore $w \geq 2^{-N}$ always, and so we have the following theorem.  

\paragraph{Theorem.} For any Heisenberg antiferromagnet in the form of \eqref{eq:HeisenbergHGeneral}, the expectation of $H$ with respect to the state
$$|M_B\rangle = \frac{(-H)^B|+^N\rangle}{\sqrt{ \langle +^N |(-H)^{2B}|+^N \rangle}}$$
satisfies $|\langle M_B | H | M_B \rangle - E_0| \leq \epsilon$ whenever
$$B \geq \frac{\|H\|}{2 \epsilon} \left ( 2 N + \ln \left (2 \|H\| \epsilon^{-1}  \right) \right). $$

\section{Observable Estimation}\label{sec:ObsAndEnergyEst}

Let $O$ be any local observable that can be decomposed into a sum of the two dual operators of the loop representation (e.g. ${O} = (\openone + X) $ or ${O} = 2 h_{ij} = I_{ij} +S_{ij}$). Define the operator generally as ${O} = W_a+W_b$.  We can estimate the expectation $\langle {O} \rangle $ using samples from the distribution
$$ \pi(x) = \frac{2^{L(x)}}{\mathcal Z} \quad , \quad \mathcal Z = \sum_{x\in \Omega} 2^{L(x)}$$
as follows.  Note that by definition we have:
\begin{align}\label{eq:SSElocalObs}
 \langle {O} \rangle \cong \frac{\sum_{\sigma_L,\sigma_R,W\in \mathfrak S}\langle \sigma_L |\prod_{t=B+1}^{2B} W(t) {O} \prod_{t=1}^{B} W(t) |\sigma_R\rangle }{\sum_{\sigma_L,\sigma_R,W\in\mathfrak S}\langle \sigma_L | \prod_{t=1}^{2B} W(t) |\sigma_R\rangle }
\end{align}
We can then expand the numerator and define it as $\mathcal Z_{ O}$:
\begin{align}
   \mathcal Z_{ O} \equiv \sum_{\sigma_L,\sigma_R,W\in \tilde{\mathfrak S}}\langle \sigma_L |\prod_{t=1}^{2B+1}W(t)  |\sigma_R\rangle 
\end{align}
where $\tilde{\mathfrak S}$ are is the set of all length $2B+1$ operator strings such that that $W(B+1)\in \{W_a,W_b\}$. This can then be incorporated into the un-directed loop representation as a sum over configurations $x_{ O}=(e_1, e_2, \cdots,e_{B},e_o,e_{B+2},\cdots,  e_{2B+1})$ where $e_o$ is always the type of operator corresponding to $ O$. We can then use the fact that a configuration $x = (e_1, \ldots,  e_{2B})$ uniquely defines a configuration $x_{ O}$ by simply inserting $e_o$ in the middle to write: 
\begin{align}
     \frac{\mathcal Z_{ O}}{\mathcal Z} = \sum_{x\in \Omega}\frac{2^{L(x_{ O})}}{\mathcal Z} = \sum_{x\in \Omega}\frac{2^{L(x)}}{\mathcal Z}2^{(L(x_{ O})-L(x))} = \sum_{x\in \Omega} \pi(x)w_{ O}(x) 
\end{align}
where $w_{ O}(x) \equiv 2^{(L(x_{ O})-L(x))}$ is the observable that returns two to the power of the difference of the number of loops that are in $x_{ O}$ and $x$. This means that using a set of samples $x_1,...,x_m$ from $\pi$, we can estimate $\langle {O}\rangle$ by finding the average  $m^{-1}\sum_{i=1}^m w_{ O}(x_i) $.

\section{Equivalence between the star graph AFM ground state and the 1D classical Potts model}\label{sec:pottsduality}

Here we describe an equivalence that holds for the special case of a star graph AFM with $N$ qubits, which allows us to map the loop representation QMC for the ground to the finite temperature partition function of a 1D classical Potts model~\cite{wu1982potts} with $N-1$ colors.   This equivalence is special to the case of the star graph and for more general ${O}(1)$-bipartite graphs it does not appear to easily extend to a finite-range classical model.

The key insight behind the equivalence is the topological observations in \Cref{lem:loopdiffbound} that sharply constrain the kinds of closed internal loops that appear in the SPE representation of the star graph ground state.  For an SPE configuration consisting of the string of operators $(x_1,..,x_{2B})$, where each $x_i$ is an AFM local term, closed internal loops appear when and only when there exists indices $i$ such that $x_i = x_{i+1}$.   Therefore the number of loops in each configuration is:
$$L(x_1,...,x_{2B}) = N + \sum_{i=1}^{2B}  \delta_{x_i,x_{i+1}}$$
where $N$ counts the number of open strings with both endpoints on a boundary, and the sum over Kronecker delta terms accounts for all of the internal loops.   Therefore the weight of each configuration is
$$2^{\textrm{ \# of loops}} = 2^N  e^{- \ln(2) E_\textrm{Potts}(x_1,...,x_{2B})}$$
were $E_\textrm{Potts}(x_1,...,x_{2B}) = -\sum_{i=1}^{2B}  \delta_{x_i,x_{i+1}}$ is the usual energy of the (ferromagnetic) Potts model.    The equivalence consists of this mapping between the SPE weight of each configuration and the Boltzmann weight of a configuration in the Potts model.   The Potts model has $2B$ sites and $N-1$ colors.

\section{Phases of matter and complexity.}\label{sec:phaseappendix}
When the graph is bipartite, the different types of ground states that an AFM Heisenberg-type model (with weights) can host are believed to be relatively limited. Here, we will focus on two major types of ``phases''. The term ``phase'' in condensed matter physics is an equivalence class over quantum states that is defined over families of increasingly large Hamiltonians with the defining property \cite{tasaki2020physics,sachdev1999quantum}. 

As explained in \Cref{sec:discussion}, the most common phase to be realized is the N\'eel phase, defined by the property that $\langle O_{\textrm{N\'eel}}^2\rangle$ converges to a nonzero constant in the $N\to\infty$ limit. 
If one reflects on how we got the loop representation in \cref{sec:looprep}, any two qubits will be completely independent if no loops contains them both simultaneously. The presence of the N\'eel order is thus equivalent to the loops in SPE being typically extensively large, at least around the center of the configuration. 
High dimensional bipartite lattices tend to realize this N\'eel order; it is proven rigorously for cubic lattices in 3-dimensions or higher, and is strongly believed with overwhelming numerical evidence in 2-dimensions \cite{dyson1978phase,kennedy1988existence}.

One way to view the N\'eel phase is that the ideal N\'eel product state $\prod^{\otimes}_{\mathcal A}|+\rangle\prod^{\otimes}_{\mathcal B}|-\rangle$ is a good approximation. 
The other contrasting phase that is commonly considered is the valence-bond-solid (VBS) phase, where the ground state can again be approximated by a product state, but in this case with singlets: $\prod^{\otimes}_{\mathcal C}(|\updown\rangle - |\downup\rangle)/\sqrt{2}$ where $\mathcal C$ is some dimer covering of the graph.
This type of state is another reasonable candidate for ground states of Heisenberg-type models, since the Heisenberg interaction is favoring singlets. 
In the loop representation, this state corresponds to having many small loops on top of the covering $\mathcal C$. 
While no pure unweighted bipartite AFM Heisenberg Hamiltonian is known to host this state as the ground state, it is known to be realized by slight extensions where we consider higher-order terms of the Heisenberg terms such as $h_{ij}h_{kl}$ \cite{sandvik07dqc}.

The two phases N\'eel and VBS turn out to be the two naturally arising phases in the context of the loop representation. 
Recall that in the simplest set up, we were just counting the number of loops and assigning a weight 2 to each loops as in \cref{eq:piEquilibrium}. 
If we extend this Monte Carlo representation to have weight $\alpha^{L(x)}$, then this corresponds to the $\mathrm{SU}(\alpha)$ generalization of the Heisenberg model for $\alpha\in\mathbb{N}$ \cite{Desai2021Resummation,kaul2015marshall}, though we can also consider more general $\alpha \in \mathbb{R}$, analogous to varying the $\delta$ parameter in the Temperley-Lieb algebra \cite{temperley2004relations}. 
In this framework, it is known that there is a phase transition from the ``extended loop phase'' (generalization/analogy of N\'eel phase) to the ``short loop phase'' (VBS phase) at $\alpha_{\mathrm c} \sim 4.5$ in 2-dimensional square lattices \cite{harada2013possibility}. 
Since the $\alpha=1$ case is a trivial sampling problem, this indicates that the $\alpha=2$ case (which we are interested in now) most likely shares the same properties such as rapid-mixing with that trivial limit. 
Furthermore, the VBS phase breaks a discrete symmetry of the lattice and has an exponentially long mixing time (See \cref{sec:vbs} for more details). 
An educated guess, inspired by the rigorous results on the 2-dimensional classical Ising model mixing time (and other models) \cite{lubetzky2012critical,levin2017markov}, would be that on the 2-dimensional square lattice, the $\mathrm{SU}(\alpha)$ generalization mixes rapidly for all $1\leq\alpha<\alpha_{\mathrm c}$ and exponentially slowly for $\alpha>\alpha_{\mathrm c}$. 
From this perspective, it is not a coincidence that the counterexample configurations in \cref{fig:badconfigs} (a-1) correspond to such VBS phases, since they are generally the bottleneck for fast-mixing.

The SU($\alpha$) generalization is regarded more its theoretical features than for describing real physical materials, yet we see here that it can provide a useful high-level perspective of the situation regarding rapid mixing of AFM Heisenberg systems.  Interestingly, all cases where a VBS phase is currently known to be manifested as a ground state, either requires $\alpha >2 $ or other deviations from our setup. 
To name a few, we can have non-bipartite AFM interaction \cite{maj69nex} 
\begin{equation}\label{eq:MajumdarGhosh}
    H= -\sum_{i} h_{i,i+1} -\frac{1}{2}\sum_{i} h_{i,i+2}
\end{equation}
that has a two-fold degenerate VBS ground state that corresponds to the two configurations in \cref{fig:badconfigs} (a-1), or have higher order terms \cite{yang2020decon,tang2011method,Sanyal2011vacancy} 
\begin{equation}\label{eq:Q2chain}
    H= -\sum_{i} h_{i,i+1}h_{i+2,i+3}
\end{equation}
where again the ground state is two-fold degenerate in the same manner. 
While the former case will completely break the stoquastifying transformation \cref{eq:stoquastic_bipartite-rotation} and thus will not allow QMC sampling at all without a drastically different approach, 
the latter with higher order $h$ terms are easily implementable with standard SSE QMC \cite{sandvik07dqc} and even with our QPE formulation.  
All it takes is to allow multiple edges (instead of just one) for a single term in \cref{eq:SSE-loop-state-space}, i.e., to include  $e_i\in\mathcal{E}\times\mathcal E$ corresponding to terms like $h_{ij}h_{kl}$ in the Hamiltonian.

These observations imply that whatever strategy/method we use to show rapid-mixing for general bipartite {\sf QuantumMaxCut} must be appropriately sensitive to the above details: $\alpha$ of the loop weight, the number of edges in a single operator, and the nonexistence of the VBS phase in a pure 2-body interacting Heisenberg model. 
The fact that the arguments we have utilized did not rely on the above details, which therefore reflects a specific limitation of the canonical paths we employed. 
One possible way to circumvent these difficulties is to {\it assume} the existence of the N\'eel order (or the absence of the VBS order), and show rapid mixing of QMC. Even with such assumptions, the proof of rapid-mixing will involve deep and novel understandings of the nature of the N\'eel phase. 

\section{Valence-Bond-Solid (VBS) phase, criticality, and rapid-mixing}\label{sec:vbs}
\vspace{-0.5cm}
\begin{figure}[h]
    \includegraphics[width=10cm]{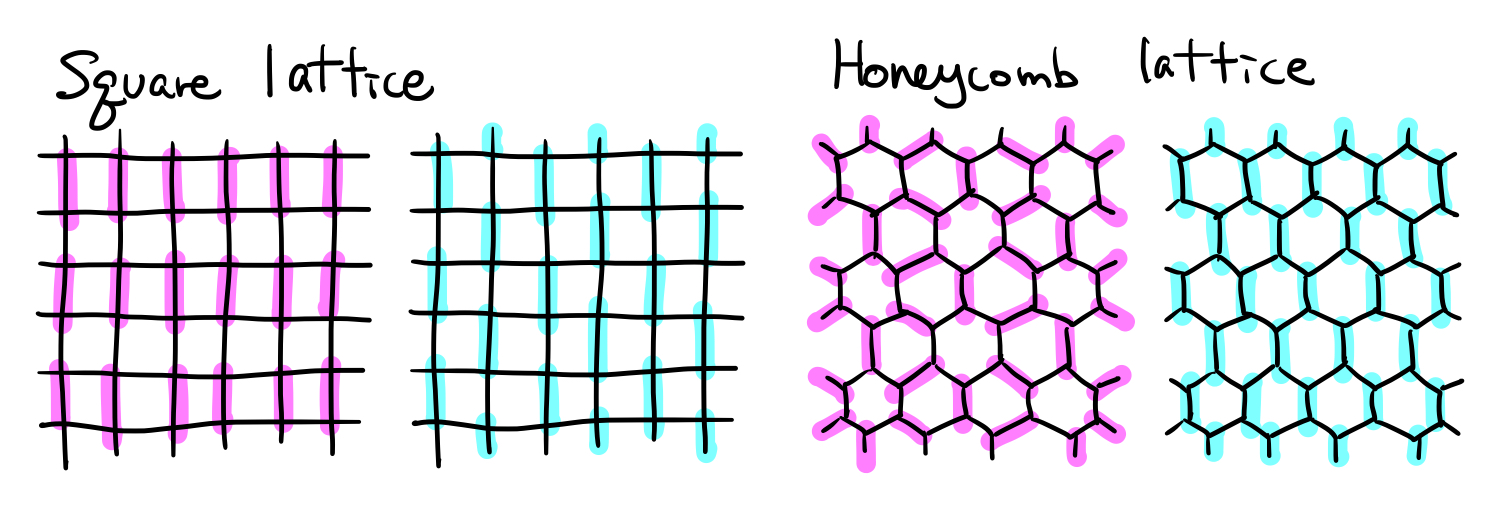}
    \caption{Periodic perfect matchings for the square and honeycomb lattices. The pink ones (both called {\it columnar} VBS) have larger fluctuation than the blue ones (both called {\it staggered} VBS), benefiniting from the entropy effect resulting to become the canonical VBS realized in those lattices. }
    \label{fig:VBSs}
\end{figure}

Here, we briefly introduce the basics for the VBS phase, which is slightly more complicated than the N\'eel phase. 
As we explained in the discussion section, the VBS phase can be thought of as a phase where short loops in the QMC configuration are dominant, and is naturally realized if we consider $\mathrm{SU}(\alpha)$ generalizations with large enough $\alpha$. This is equivalent to having loop factors $\alpha$ instead of 2 in the normal spin-$1/2$ case. 

When the tendency of favoring larger numbers of loops is strong enough, the configurations will typically look similar to those where you have the maximum amount of loops, i.e., a VBS configuration. 
More precisely, in an ideal VBS configuration, a perfect matching of the graph $\mathcal G$ is chosen, and all bridges in the loop representation exist only on those edges in the perfect matching. This will give you essentially the theoretical maximum number of loops $L_{\max}\sim 2B$. This was depicted in \cref{fig:badconfigs} (a-1). 

One may think that in a square lattice there are exponentially many perfect matchings, so the VBS phase should also have equally vast degeneracy. This is not the case, since usually entropy effects remain and there would statistically be some ratio of bridges placed on edges that are {\it not} on the selected perfect matchings. Depending on the VBS pattern (= the perfect matching), this entropy effect varies largely, and only the perfect matching(s) that benefits the most from the entropy effect are the true ground states, as shown in \cref{fig:VBSs}. 

For example, in the 1-dimensional case (the cycle graph) there is only one type of perfect matching (simply choose every other edge), and two ways to realize it: choose between even or odd edges, as in \cref{fig:badconfigs} (a-1) $x$ or $y$. Either way, the order will be captured by the following order parameter 
\begin{equation}
    O_{\mathrm{VBS}}=\frac{1}{N}\sum_{i} (-1)^i h_{i,i+1}, 
\end{equation}
where $h$ is the singlet projector as in \cref{eq:MajumdarGhosh}, \eqref{eq:Q2chain}, or \eqref{eq:HeisenbergHGeneral}. 
Similarly to the N\'eel phase, the VBS phase is technically defined as a non-vanishing value of $\langle O_{\mathrm{VBS}}^2\rangle$ in the thermodynamic limit. 
This definition needs to be modified accordingly to the expected VBS form for higher dimensional lattices.

A crucial point in our discussion in \cref{sec:discussion} was that the VBS phase spontaneously breaks a discrete symmetry of the lattice. In the case of the 1-dimensional chain we explained above, this symmetry corresponded to $\mathbb{Z}_2$ with lattice translation (or spatial reflection with respect to a vertex). In the 2-dimensional square lattice case with columnar VBS in \cref{fig:VBSs}, it becomes the four-fold degeneracy corresponding to $\mathbb{Z}_4$ rotations or $\mathbb{Z}_2\times\mathbb{Z}_2$ reflections \cite{takahashi20vbs}. Since the loop configuration directly reflects the choice of the VBS state (which perfect matching pattern), we must move almost all the bridges in the SPE QMC configuration to move from one state to another, but separated by an exponentially unlikely configuration with a $\textrm{poly}(N)$-size domain wall. This is the reason why discrete symmetry breaking phases are expected to mix exponentially slowly, just as in the Ising model in the ordered phase \cite{lubetzky2012critical}. 

Because the discrete symmetry breaking is key for the slow mixing argument, the same logic does not hold for Hamiltonians with very similar ground states to VBS states as long as they do not spontaneously break discrete symmetry.  
For example, we can consider any lattice e.g. as in \cref{fig:VBSs}, and have 
\begin{equation}
    H = -\sum_{(i,j)\in\mathcal E} (1\pm\Delta) h_{i,j}, 
\end{equation}
where we choose $\pm\Delta$ so that the bonds corresponding to the target VBS pattern is strong. In this way, the ground state is naturally realized as a perfect matching configuration that very much resembles a VBS pattern, but is not spontaneously breaking the symmetry (the symmetry is already explicitly broken in the Hamiltonian). 
In these ground states, usually called as dimer(ized) states, the Monte Carlo configurations {\it do not} need to move between the different symmetry sectors, and thus has no exponential barrier as in the VBS phase. 

One thing to note however, is that in typical cases such as the columnar pattern for the square lattice, the phase transition from N\'eel-to-dimer occurs at some nonzero $\Delta$ value. Currently, all of the observed phase transition of this type are numerically confirmed to be critical (equivalently, second-order and/or continuous) phase transitions, which the slowdown of mixing is only polynomial \cite{troyer1997critical,sandvik06dimer}. A discovery of a first order (discontinuous) N\'eel-to-dimer transition (exponentially slow mixing at the transition point because of $O(N)$ distant configurations with close weights) will therefore be a very strong evidence against our conjecture we argued in \cref{sec:discussion} that (weighted) bipartite {\sf QuantumMaxCut} is in {\sf BPP} through quantum Monte Carlo sampling.

\end{document}